\documentclass{article}

\usepackage{parskip}

\usepackage{amsmath,mathtools}

\usepackage[margin = 1.5in]{geometry}

\usepackage{times}
\usepackage{graphicx} % more modern
\usepackage{subfigure}
\usepackage{epstopdf}

\usepackage[numbers]{natbib}
\usepackage{amsmath}
\usepackage{amssymb}
\usepackage{amsfonts}

\usepackage{color}
\usepackage{booktabs}

\usepackage{algorithm}
\usepackage{algorithmic}

\usepackage{hyperref}
\usepackage[normalem]{ulem}
\usepackage{verbatim} 

\newtheorem{theorem}{Theorem}
\newtheorem{lemma}[theorem]{Lemma}
\newtheorem{corollary}[theorem]{Corollary}

\newtheorem{problem}{Problem}

\newenvironment{proof}{\setlength\parindent{0pt}{\bf Proof.    }}{\hfill\rule{2mm}{2mm}}

\def\R{\mathbb{R}}

\begin{document}

\title{Approximating Spectral Sums of Large-scale Matrices using Stochastic Chebyshev Approximations\footnote{{This article is partially based on preliminary results published in the proceeding of the 32nd International Conference on Machine Learning (ICML 2015).}}}

% Haim: "Stochastic Chebysev" is a bit weird since the Chebyshev part is not stochastic 

%Approximating Large-scale Spectral Functions through Stochastic Chebyshev
%Approximating Positive Definite Matrix Functions through Stochastic Chebyshev Expansions}

\author{
Insu Han \thanks{School of Electrical Engineering, Korea Advanced Institute of Science and Technology, Korea. 
Emails: hawki17@kaist.ac.kr
} \and 
Dmitry Malioutov \thanks{Business Analytics and Mathematical Sciences, IBM Research, Yorktown Heights, NY, USA. 
Email: dmalioutov@us.ibm.com
} \and 
Haim Avron \thanks{Department of Applied Mathematics, Tel Aviv University. 
Email: haimav@post.tau.ac.il
} \and
Jinwoo Shin \thanks{School of Electrical Engineering, Korea Advanced Institute of Science and Technology, Korea. 
Email: jinwoos@kaist.ac.kr
}
}

\maketitle

\begin{abstract} 
Computation of the trace of a matrix function plays an important role in many scientific computing applications,
including applications in machine learning, computational physics (e.g., lattice quantum chromodynamics), network analysis and computational biology (e.g., protein folding), just to name a few application areas. 
We propose a linear-time randomized algorithm for approximating the trace of 
matrix functions of large symmetric matrices. Our algorithm is based on coupling 
function approximation using Chebyshev interpolation with stochastic trace estimators (Hutchinson's
method), and as such requires only implicit access to the matrix, in the form of a function that
maps a vector to the product of the matrix and the vector. We provide rigorous approximation error in terms
of the extremal eigenvalue of the input matrix, and the Bernstein ellipse that corresponds to the function
at hand. Based on our general scheme, we provide algorithms with provable guarantees for important 
matrix computations, including log-determinant, trace of matrix inverse, Estrada index, Schatten $p$-norm,
and %a novel algorithm for 
testing
positive definiteness. We experimentally evaluate our algorithm and demonstrate its effectiveness
on matrices with tens of millions dimensions. 
\end{abstract} 

\section{Introduction}

Given a symmetric matrix  $A \in \mathbb R^{d\times d}$ and function {$f:\mathbb R\rightarrow \mathbb R$}, we study how to efficiently compute
\begin{equation}
\Sigma_f(A) = {\tt tr}(f(A)) = \sum_{i=1}^d f( \lambda_i),\label{eq:goal}
\end{equation} 
where $\lambda_1 , \dots , \lambda_d$ are eigenvalues of $A$. We refer to such sums
as {\em spectral sums}. Spectral sums depend only on the eigenvalues of $A$ and so
they are {\em spectral functions}, although not every spectral function is a spectral sum.
Nevertheless, the class of spectral sums is rich and includes useful spectral functions.
For example, if $A$ is also positive definite then $\Sigma_{\log}(A)=\log\det(A)$, i.e. the log-determinant
of $A$.

Indeed, there are many real-world applications in which spectral sums play an important role.
For example, the log-determinant appears ubiquitously in machine learning applications
including Gaussian graphical and Gaussian process models \cite{rue_GMRF, rasmussen_GP, dempster_cov_sel}, 
partition functions of discrete graphical models \cite{ma2013estimating},
minimum-volume ellipsoids \cite{MVE_rousseuw}, metric learning and kernel learning \cite{Dhillon_metric_learning}.
The trace of the matrix inverse ($\Sigma_f(A)$ for $f(x)=1/x$)  is frequently computed for the covariance matrix in uncertainty quantification~\cite{dashti2011uncertainty,kalantzis2013accelerating}
and lattice quantum chromodynamics~\cite{stathopoulos2013hierarchical}.
The Estrada index ($\Sigma_{\exp}(A)$) has been initially developed for topological index of protein folding
in the study of protein functions and protein-ligand interactions \cite{estrada2000char, de2007estimating},
and currently it appears in numerous other applications, e.g., statistical thermodynamics \cite{estrada2007statistical,estrada2008atom}, information theory \cite{carbo2008smooth} and network theory \cite{estrada2005spectral,estrada2007topological} ; see Gutman et al.~\cite{gutman2011estrada} for more applications.
The Schatten $p$-norm {($\Sigma_f(A^\top A)^{1/p}$} for $f(x)=x^{p/2}$ {for $p \geq 1$} ) has been applied to recover low-rank matrix \cite{nie2012low} and sparse MRI reconstruction \cite{majumdar2011algorithm}.
%The positive definite testing ($\Sigma_f(A)$ for $f(x)=1$ for $x\leq 0$ and $f(x)=0$ for $x> 0$)

The computation of the aforementioned spectral sums for large-scale matrices is a challenging task.
For example, the standard method for computing the log-determinant uses the 
Cholesky decomposition (if $A = L L^T$ is a Cholesky decomposition, then $\log \det (A) = 2 \sum_i \log L_{ii}$).
In general, the computational complexity of Cholesky decomposition is cubic with respect to the number of
variables, i.e. $O(d^3)$.
For large-scale applications involving more than tens of thousands
of dimensions, this is obviously not feasible. 
If the matrix is sparse, one might try to take advantage of sparse decompositions. As long
as the amount of fill-in during the factorizations is not too big, a substantial improvement
in running time can be expected. Nevertheless, the worst case still requires $\Theta(d^3)$. In particular,
if the sparsity structure of $A$ is random-like, as is common in several of the aforementioned applications,
then little improvement can be expected with sparse methods.

Our aim is to design an efficient algorithm that is able to
compute accurate {approximations to  spectral sums for matrices with {\em tens of
millions} of variables.}

\subsection{Contributions}
We propose a randomized algorithm for approximating spectral sums based on a combination of stochastic trace-estimators 
and Chebyshev interpolation. Our algorithm first computes the coefficients of a  Chebyshev approximation of $f$. This immediately
leads to an approximation of the spectral sum as the trace of power series of the input matrix. 
We then use a stochastic trace-estimator to estimate this trace. In particular, we use {\em Hutchinson's method}~\cite{hutchinson1989stochastic}.  

One appealing aspect of Hutchinson's method is that
it does not require an explicit representation of the input matrix; Hutchinson's method 
requires only an implicit representation of the matrix as an operation that maps a vector 
to the product of the matrix with the vector. In fact, this property 
is inherited by our algorithm to its entirety: our algorithm only needs access to 
an implicit representation of the matrix as an operation that maps a vector 
to the product of the matrix with the vector. In accordance, we measure the complexity of our 
algorithm in terms of the number of matrix-vector products that it requires.
We establish rigorous bounds on the number of matrix-vector products for attaining 
a $\varepsilon$-multiplicative approximation of the spectral sum based on $\varepsilon$, the 
failure probability and the range of the function over its Bernstein ellipse 
(see Theorem \ref{main} for details). %in its analytic region.
In particular, Theorem~\ref{main}
implies that if the range is $\Theta(1)$, then the algorithm provides
$\varepsilon$-multiplicative approximation guarantee using a constant amount of matrix-vector products
for any constant $\varepsilon>0$ and constant failure probability.

The overall time complexity of our algorithm is $O(t \cdot \| A\|_{\tt mv})$ where $t$ is the 
number of matrix-vector products (as established by our analysis) and $\| A\|_{\tt mv}$ is the cost
of multiplying $A$ by a vector. One overall assumption is that matrix-vector products can be computed 
efficiently, i.e. $\| A\|_{\tt mv}$ is small. For example, if $A$ is sparse then 
$\| A\|_{\tt mv} = O({\tt nnz}(A))$, i.e., the number of non-zero entries in $A$. 
Other cases that admit fast matrix-vector products are 
low-rank matrices  (which allow fast multiplication by factorization), % matrices,
or Fourier (or Hadamard, Walsh, Toeplitz) matrices using the fast Fourier transform.
The proposed algorithm is also very easy to parallelize.

We then proceed to discuss applications of the proposed algorithm. We give rigorous bounds for
using our algorithm for approximating the log-determinant, trace of the inverse of a matrix, the 
Estrada index, and the Schatten $p$-norm. These correspond to continuous functions
$f(x) = \log x$, $f(x)=1/x$, $f(x) =\exp(x)$ and $f(x)=x^{p/2}$, respectively. 
We also use our algorithm to construct a novel algorithm for testing positive definiteness
in the property testing framework. Our algorithm, which is based on approximating the spectral sum for 
$1-\mbox{sign}(x)$, is able to test positive definiteness of a matrix with a sublinear (in matrix
size) number of matrix-vector products. 

Our experiments show that our proposed algorithm %when applied to approximate the log-determinant 
is orders of magnitude faster than the standard methods for
sparse matrices and provides approximations with less than $1\%$ error for
the examples we consider. It
can also solve problems of tens of millions dimension in a few minutes on our single
commodity computer with 32 GB memory. Furthermore, as reported in our experimental results, 
it achieves much better accuracy compared to a similar approach
based on Talyor expansions~\cite{leithead_GP_stoch}, while both have similar running times.
In addition, it outperforms the recent method based on Cauchy integral formula~\cite{aune2014_GP} in both running time
and accuracy.\footnote{{Aune et al.'s method~\cite{aune2014_GP} is implemented in the SHOGUN machine learning toolbox, http://www.shogun-toolbox.org.}}
%, where the method
%is implemented in the SHOGUN machine learning toolbox\footnote{http://www.shogun-toolbox.org/}.
The proposed algorithm is also very easy to parallelize and hence has a potential to handle
{even larger problems.}
For example, the
Schur method was used as a part of QUIC algorithm for sparse
inverse covariance estimation with over million variables~\cite{hsieh2013big}, hence our log-determinant algorithm could
be used to further improve its speed and scale.

%\vspace{0.1in}
%\noindent{\bf Related work.} 
\subsection{Related Work}

Bai et al.~\cite{bai96largescale} were the first to consider the problem of approximating spectral
sums, and its specific use for approximating the log-determinant and the trace of the 
matrix inverse. Like our method, their method combines stochastic trace estimation with approximation of 
bilinear forms. However, their method for approximating bilinear forms is fundamentally different 
than our method and is based on a Gauss-type quadrature of a Riemann-Stieltjes integral. They do not 
provide rigorous bounds for the bilinear form approximation. In addition, recent progress on analyzing
stochastic trace estimation~\cite{avron2011randomized,roosta2013improved} allow us to provide 
rigorous bounds for the entire procedure.

Since then, several authors considered the use of stochastic trace estimators to compute
certain spectral sums. Bekas et al.~\cite{bekas2007diag} and Malioutov et al.~\cite{malioutov2006low} 
consider the problem of computing the diagonal of a matrix or of the matrix inverse. 
Saad et al.~\cite{eigenvalue_histograms_Saad} use polynomial approximations and rational approximations
of high-pass filter to count the number of eigenvalues in an input interval. They do not provide
rigorous bounds. Stein et al.~\cite{stein2013stochastic} use stochastic approximations of score functions
to learn large-scale Gaussian processes. 

Approximation of the log-determinant in particular has received considerable
treatment in the literature. Pace and LeSage~\cite{pace2004chebyshev} use both Taylor and Chebyshev
based approximation to the logarithm function to design an algorithm for log-determinant approximation, but do not use stochastic trace estimation. Their method is determistic, can entertain only low-degree approximations, and
has no rigorous bounds. Zhang and Leithead~\cite{leithead_GP_stoch} consider the problem
of approximating the log-determinant in the setting of Gaussian process parameter learning. They use Taylor 
expansion in conjunction with stochastic trace estimators, and propose novel error compensation methods. 
They do not provide rigorous bounds as we provide for our method. 
Boutsidis et al.~\cite{boutsidis2015randomized} use a similar scheme based on Taylor expansion for 
approximating the log-determinant, and do provide rigorous bounds. Nevertheless, our experiments demonstrate
that our Chebyshev interpolation based method provides superior accuracy. Aune et al.~\cite{aune2014_GP} 
approximate the log-determinant using a Cauchy integral formula. Their method requires the multiple use of a 
Krylov-subspace linear system solver, so their method is rather expensive. Furthermore, no rigorous bounds
are provided. 

Computation of the trace of the matrix inverse has also been researched extensively. One recent example
is the work of Wu et al.~\cite{wu15interpolating} use a combination of stochastic trace estimation 
and interpolating an approximate inverse. In another example, Chen~\cite{chen15accurate} considers how 
accurately should  linear systems be solved when stochastic trace estimators are used to approximate the trace of the inverse. 

To summarize, the main novelty of our work is combining Chebyshev interpolation with Hutchinson's trace
estimator, which allows to design an highly effective linear-time algorithm with rigorous approximation guarantees
for general spectral sums.

\subsection{Organization}
The structure of the paper is as follows. We introduce the necessary background
in Section \ref{sec:back}. 
Section \ref{sec:main} provides
the description of our algorithm with approximation
guarantees, and  its applications to 
the log-determinant, the trace of matrix inverse, the Estrada index, the Schatten $p$-norm and testing positive definiteness are described in Section \ref{sec:application}.
%Section \ref{sec:proof} provides the proof of the main theorem, and 
We report experimental results in Section \ref{sec:exp}.

\section{Preliminaries} \label{sec:back}

Throughout the paper, $A \in \mathbb{R}^{d \times d}$ is a symmetric matrix with eigenvalues $\lambda_1, \dots, \lambda_d \in \mathbb{R}$ {and $I_d$ is the $d$-dimensional identity matrix}. We use ${\tt tr}(\cdot)$ to denote the trace of the matrix. 
%\textcolor{red}{
We denote the Schatten $p$-norm by $\|\cdot\|_{(p)}$,
and the induced matrix $p$-norm by $\|\cdot\|_p$ (for $p=1,2,\infty$) .
%}
We also use $\lambda_{\min}(A)$ and $\lambda_{\max}(A)$ to denote the smallest and largest eigenvalue of  $A$. In particular, we assume that an interval $[a,b]$ in which contains all of $A$'s eigenvalues is given. In some cases, such bounds are known a-priori due to properties of the downstream use (e.g., the application considered in Section~\ref{sec:gmrf}). In others, a crude bound like  $a = - \|A\|_{\infty}$ and $b= \|A\|_{\infty}$ or via Gershgorin's Circle Theorem \cite[Section 7.2]{golub2012matrix} might be obtained.
For some functions, our algorithm has additional requirements on $a$ and $b$ (e.g. for log-determinant, we need $a > 0$).

%or one can run the power iteration \cite{ipsen1997computing} to estimate a better bound.
%On the other hand, the lower bound of eigenvalues is generally not easy to obtain, except for special cases
%including random matrices \cite{tao2009inverse, tao2010random}
%and diagonal-dominant matrices \cite{gershgorin1931uber, moravca2008bounds}.
%It is easy to obtain in the problem of counting spanning trees we studied in Section \ref{sec:spanning},
%and it is explicitly given as a parameter in many machine learning log-determinant applications
%\cite{wainwright2006log}.
%In general, one can use the inverse power iteration \cite{ipsen1997computing} to estimate it. 
%and assume that $a$ and $b$ are given.%, although it would be first necessary to estimate them.
%Although, in practice, this interval would be required to estimate first,
%we assume that $a$ and $b$ are given.
%However, in practice, 
%%% conditions for function f
%% Function setting
%We make assumption about the function $f$ as the followings:
%Under the assumptions, one can see that $f$ is always positive or negative in the interval $[a,b]$.
%In the rest of this section, we describe the preliminaries for our approach to approximate matrix functions.
% under
%that satiesfied with 
%the above assumptions.

Our approach combines two techniques, which we discuss in detail in the next two subsections:
(a) designing polynomial expansion for given function via Chebyshev interpolation \cite{mason2002chebyshev} 
and (b) approximating the trace of matrix via Monte Carlo methods \cite{hutchinson1989stochastic}.

\subsection{Function Approximation using Chebyshev Interpolation}\label{sec:cheby}
Chebyshev interpolation approximates an analytic function by interpolating the function at the 
Chebyshev nodes using a polynomial. Conveniently, the interpolation can be expressed in terms
of basis of Chebyshev polynomials. Specifically, the Chebyshev interpolation $p_n$ of degree $n$ for a given function $f : [-1,1] \rightarrow \mathbb{R}$ is given by (see Mason and Handscomb~\cite{mason2002chebyshev}):
\begin{equation}
f(x) \approx p_n(x) = \sum_{j=0}^n c_j T_j(x) \label{eq:cheb0def}
\end{equation}
where the coefficient $c_j$, the $j$-th Chebyshev polynomial $T_j(x)$ and Chebyshev nodes $\{x_k\}^n_{k=0}$ are defined as
\begin{align}
&c_j
= \begin{dcases}
      \frac{1}{n+1}  \sum_{k=0}^n f(x_k) \ T_0(x_k) & \text{if $\ j=0$}\\
      \frac{2}{n+1}  \sum_{k=0}^n f(x_k) \ T_j(x_k) & \text{otherwise} 
      \end{dcases} \label{eq:cdef}\\
      \nonumber \\
&T_0(x) = 1, T_1(x) = x \nonumber \\
&T_{j+1}(x) = 2 x T_j (x) - T_{j-1} (x) \label{eq:chebreq} \qquad \text{for $\ j \ge 1$} \\
&x_k = \cos \Big( \frac{ \pi (k + 1/2 )}{ n+1} \Big)\,. \nonumber 
\end{align}

Chebyshev interpolation better approximates the functions as the degree $n$ increases.
In particular, the following error bound is known~\cite{berrut2004barycentric, xiang2010error}.
%%% p_n converge to f(x)
\begin{theorem}\label{converge_unit}
Suppose $f$ is analytic function with $\left| f ( z ) \right| \le U$ in the region bounded by the so-called
Bernstein ellipse
with foci $+1,-1$ and sum of major and minor semi-axis lengths equals to $\rho > 1$. Let $p_n$ denote the degree $n$ Chebyshev interpolant of $f$ as defined by equations~\eqref{eq:cheb0def},~\eqref{eq:cdef} and~\eqref{eq:chebreq}. We have,
%\vspace{-0.1in}
\begin{align*}
\max_{x \in [-1,1]} \left| f(x) - p_n(x) \right| \le  \frac{4U}{\left( \rho -1 \right) \rho^{n}}.
\end{align*}
\end{theorem}

%%% Linear mapping from [-1,1] to [a,b]
The interpolation scheme described so far assumed a domain of $[-1,1]$. To allow a more general 
domain of $[a,b]$ one can use the linear mapping $g(x) = \frac{b-a}{2}x+\frac{b+a}{2}$ to map $[-1,1]$ 
to $[a,b]$. Thus, $f \circ g$ is a function on $[-1,1]$ which can be approximated using the scheme
above. The approximation to $f$ is then ${\widetilde p}_n= p_n \circ g^{-1}$ where $p_n$ is the approximation to 
$f \circ g$. Note that $\widetilde p_n$ is
a polynomial with degree $n$ as well. In particular, we have the following  approximation scheme for a general $f : [a,b] \rightarrow \mathbb{R}$:
\begin{equation}
f(x) \approx {\widetilde p}_n(x) = \sum_{j=0}^n {\widetilde c}_j T_j\left(\frac{2}{b-a}x - \frac{b+a}{b-a}\right) \label{eq:chebdef}
\end{equation}
where the coefficient $\widetilde c_j$ are defined as
\begin{align}
&{\widetilde c}_j
= \begin{dcases}
      \frac{1}{n+1}  \sum_{k=0}^n f\left(\frac{b-a}{2}x_k+\frac{b+a}{2}\right) \ T_0(x_k) & \text{if $\ j=0$}\\
      \frac{2}{n+1}  \sum_{k=0}^n f\left(\frac{b-a}{2}x_k+\frac{b+a}{2}\right) \ T_j(x_k) & \text{otherwise}
      \end{dcases}  \label{eq:ctildedef}
\end{align}

The following is simple corollary of Theorem~\ref{converge_unit}.
\begin{corollary}\label{converge}
Suppose that $a,b\in \mathbb{R}$ with $a <b$. Suppose $f$ is analytic function with $\left| f ( \frac{b-a}{2}z + \frac{b+a}{2} ) \right| \le U$ in the region bounded by the ellipse with foci $+1,-1$ and sum of major and minor semi-axis lengths equals to $\rho > 1$. Let $\widetilde{p}_n$ denote the degree $n$ Chebyshev interpolant of $f$ on $[a,b]$ as defined by equations~\eqref{eq:chebdef},~\eqref{eq:ctildedef} and~\eqref{eq:chebreq}. We have,
%\vspace{-0.1in}
\begin{align*}
\max_{x \in [a,b]} \left| f(x) - \widetilde{p}_n(x) \right| \le  \frac{4U}{\left( \rho -1 \right) \rho^{n}}.
\end{align*}
\end{corollary}

\begin{proof}
Follows immediately from Theorem~\ref{converge_unit} and observing that for $g(x) = \frac{b-a}{2}x+\frac{b+a}{2}$
we have 
$$
\max_{x \in [-1,1]} \left| (f \circ g)(x) - p_n \left( x\right)\right| 
= 
\max_{x \in [a,b]} \left| f(x) - \widetilde p_n \left( x\right)\right|. 
$$
\end{proof}

Chebyshev interpolation for scalar functions can be naturally generalized to matrix functions \cite{higham2008functions}.
Using the Chebyshev interpolation $\widetilde p_n$ for function $f$, 
we obtain the following approximation formula:
\begin{align*}
\Sigma_f (A) &= \sum_{i=1}^d f( \lambda_i) \approx \sum_{i=1}^d \widetilde p_n(\lambda_i) = \sum_{i=1}^d \sum_{j=0}^n \widetilde c_j  T_j\left(\frac{2}{b-a}\lambda_i - \frac{b+a}{b-a}\right) \\
&= \sum_{j=0}^n \widetilde c_j \sum_{i=1}^d T_j\left(\frac{2}{b-a}\lambda_i - \frac{b+a}{b-a}\right) 
= \sum_{j=0}^n \widetilde c_j {\tt tr}\left(T_j \left(\frac{2}{b-a}A - \frac{b+a}{b-a}I_d\right)\right) \\
&={\tt tr}\left( \sum^n_{j=0} \widetilde c_j  T_j \left(\frac{2}{b-a}A - \frac{b+a}{b-a}I_d\right) \right)
\end{align*}
where 
the equality before the last follows from the fact that $\sum_{i=1}^d p(\lambda_i) = {\tt tr}(p(A))$ for
any polynomial $p$, and the last equality from the linearity of the trace operation.

We remark that other polynomial approximations, e.g. Taylor, can also be used.  
However, it known that Chebyshev interpolation, in addition to its simplicity, is nearly optimal~\cite{trefethen2012atap}
with respect to the $\infty$-norm so is well-suited for our uses.
%\textcolor{red}{It is often much better than the best approximation in term of $2$-norm.}

%to
%approximate {spectral functions.} %log-determinants. 
%This paper focuses on the Chebyshev approximation, where
%Chevyshev approximation is known to be an optimal polynomial interpolation that minimize the $\ell_\infty$-error \cite{de2012mathematics}.

\subsection{Stochastic Trace Estimation (Hutchinson's Method)} \label{sec:trace}

The main challenge in utilizing the approximation formula at the end of the last subsection is how 
to compute 
$$
{\tt tr}\left( \sum^n_{j=0} \widetilde c_j  T_j \left(\frac{2}{b-a}A - \frac{b+a}{b-a}I_d\right) \right)
$$
without actually computing the matrix involved (since the latter 
is expensive to compute). In this paper we turn to the stochastic trace estimation method.
In essence, it is a Monte-Carlo approach: to estimate the trace of an arbitrary matrix $B$, first a random vector  $\mathbf{z}$ is drawn from some fixed distribution, such that
the expectation of $\mathbf{z}^\top B \mathbf{z}$ is equal to the trace of $B$. By sampling
$m$ such i.i.d. random vectors, and averaging we obtain an estimate of ${\tt tr}(B)$. 
% Insu : I think it is better write with upper index such as v^(1) ... v^(m) because one can generate 
% random matrix V = [v^(1), v^(2), ... , v^(m)] \in {-1,1}^(d x m) at first time. 
% Then, v^(i) means i-th column of random matrix V.
%Suppose that 
Namely, given random vectors  {$\mathbf{v}^{(1)},\dots,\mathbf{v}^{(m)}$}, the estimator is 
$$
{\tt tr}_m (B) = \frac1m \sum_{i=1}^m \mathbf{v}^{(i)\top} B \mathbf{v}^{(i)} \,.
$$

%\textcolor{red}{
%It is known that 
Random vectors can be used for the above trace estimator
as long as they have zero means and unit covariances \cite{hutchinson1989stochastic}.
Examples include those from 
%Such random vectors are chosen from 
%i.i.d.\ 
Gaussian (normal) distribution %, random unit vectors
and Rademacher distribution. % whose entries are uniformly drawn from $\{ -1,+1 \}$
%trace estimator using i.i.d. random vectors 
%~\cite{hutchinson1989stochastic, avron2011randomized}.
The latter sampling entries uniformly at random from $\{ -1,+1 \}$
%with equal probability
is known to have the smallest variance among such Monte-Carlo methods \cite{avron2011randomized}.
This is called as the {\em Hutchinson's estimator} and
%and has found use in many applications~\cite{avron2010counting, hutchinson1989stochastic, aravkin2012robust}.
%We show Hutchinson trace estimator is more accurate practically for the proposed algorithm in section \ref{sec:exp}.
%}
%Formally, Hutchinson trace estimator ${\tt tr}_m(B)$ 
%is known to {
satisfies the following equalities:
%and this satisfies the following:
\begin{eqnarray*}
\mathbf{E}\left[{\tt tr}_m \left( B\right) \right] & = &{\tt tr}\left( B \right) \\
\mathbf{Var}\left[ {\tt tr}_m \left( B\right) \right] & = & \frac2{m} \left( \| B \|_{F}^2 - \sum_{i=1}^d B_{i,i}^2 \right)
\end{eqnarray*}
However, $(\varepsilon,{\zeta})$-bounds, as introduced by Avron et al.~\cite{avron2011randomized}, are more appropriate for our needs. Specifically, we use the following bound due to Roosta-Khorasani and Ascher~\cite{roosta2013improved}.
\begin{theorem}\label{thm:trace} 
Let $B \in \mathbb{R}^{d \times d}$ be a positive (or negative) semi-definite matrix.
%Assume \textit{trace estimator} ${\tt trace}_m(B)$ uses Hutchinson method. Then,
Given $\varepsilon , \zeta \in (0,1)$, 
\begin{align*}
\Pr \left[ \left| {\tt tr}_m(B) - {\tt tr}(B) \right| \le \varepsilon \left| {\tt tr}(B)\right| \right] \ge 1 - {\zeta}
\end{align*}
holds if sampling number $m$ is larger than $6 \varepsilon^{-2}  \log \left( \frac2{\zeta} \right)$.
\end{theorem}

Note that computing $\mathbf{v}^{(i)\top} B \mathbf{v}^{(i)}$ requires only multiplications between a matrix and a vector, which is particularly appealing when evaluating $B$ itself is expensive, e.g., 
$$B= \sum^n_{j=0} \widetilde c_j  T_j \left(\frac{2}{b-a}A - \frac{b+a}{b-a}I_d\right)$$ 
as
in our case. In this case,
$$
\mathbf{v}^{(i)\top} B \mathbf{v}^{(i)} = \sum^n_{j=0} \widetilde c_j \mathbf{v}^{(i)\top} T_j \left(\frac{2}{b-a}A - \frac{b+a}{b-a}I_d\right)  \mathbf{v}^{(i)} = \sum^n_{j=0} \widetilde c_j \mathbf{v}^{(i)\top} \mathbf{w}_{j}^{(i)}
$$ where
$$
\mathbf{w}_{j}^{(i)} = T_{j} \left(\frac{2}{b-a}A - \frac{b+a}{b-a}I_d\right)  \mathbf{v}^{(i)}\,.
$$
The latter can be computed efficiently (using $n$ matrix-vector products with $A$) by observing that
due to equation~\eqref{eq:chebreq} we have
that
\begin{align*}
&\mathbf{w}_{0}^{(i)} = \mathbf{v}^{(i)}, \mathbf{w}_{1}^{(i)} = \left(\frac{2}{b-a}A - \frac{b+a}{b-a}I_d\right) \mathbf{w}_{0}^{(i)}\\
&\mathbf{w}_{j+1}^{(i)} = 2 \left(\frac{2}{b-a}A - \frac{b+a}{b-a}I_d\right) \mathbf{w}_{j}^{(i)} 
- \mathbf{w}_{j-1}^{(i)}\,.
\end{align*}

In order to apply Theorem~\ref{thm:trace} we need
$B$ to be positive (or negative) semi-definite. In our case $B=\widetilde p_n (A)$ so it is sufficient
for $\widetilde p_n$ to be non-negative (non-positive) on $[a,b]$. The following lemma establishes a sufficient 
condition for non-negativity of $\widetilde p_n$, and a consequence positive (negative) semi-definiteness of 
$\widetilde p_n(A)$.

\begin{lemma} \label{positive}
Suppose $f$ satisfies that $\left| f(x) \right|\ge L$ for $x \in [a,b]$.
Then, linear transformed Chebyshev approximation $\widetilde p_n(x)$ of $f(x)$ is also non-negative on $[a,b]$ if
\begin{align}
\frac{4U}{\left( \rho - 1\right) \rho^n} \le L \label{condition}
\end{align}
holds for all $n\ge1$.
\end{lemma}
\begin{proof}
From Corollary \ref{converge}, we have
\begin{align*}
\min_{[a,b]}\widetilde p_n(x) &= \min_{[a,b]} f(x) + \left( \widetilde p_n(x) - f(x) \right) \\
&\ge \min_{[a,b]} f(x) - \max_{[a,b]} \left| \widetilde p_n(x) - f(x) \right| \\
&\ge L - \frac{4U}{\left( \rho - 1\right) \rho^n} \ge 0.
\end{align*}
where the last inequality uses Corollary \ref{converge}.
\end{proof}

\section{Approximating Spectral Sums}\label{sec:main}
%\section{Approximation Algorithm for Spectral Functions}\label{sec:main}
%% Matrix setting
%We deal with matrix $A \in \mathbb{R}^{d \times d}$ whose eigenvalues lie in the interval in $[a,b]$ for $0<a<b$.
%Let $\sigma_{\min}$ and $\sigma_{\max}$ be the minimum and maximum eigenvalues of $A$.
%If $A$ is symmetric, eigenvalues of $A$ correspond to its singular values.
%Although, in practice, this interval would be required to estimate first,
%we assume that $a$ and $b$ are given.
%However, in practice, 
%%% conditions for function f
%\subsection{Algorithm}
%
\subsection{Algorithm Description}
Our algorithm brings together the components discussed in the previous section.
A pseudo-code description appears as {\bf Algorithm \ref{alg1}}.
As mentioned before, we assume that eigenvalues of $A$ are in the interval $[a,b]$ for some $b>a$.

\begin{algorithm}[tbh!]
   \caption{Trace of matrix function $f$ approximation}
\begin{algorithmic}\label{alg1}
   \STATE {\bfseries Input:} %\sout{positive semi-definite} 
   symmetric matrix $A \in \mathbb{R}^{d \times d}$ with eigenvalues in $[a,b]$,  sampling number $m$ and polynomial degree $n$
   \STATE {\bfseries Initialize:} $\Gamma \leftarrow 0$
   %, $\widetilde A \leftarrow \frac2{b-a} A - \frac{b+a}{b-a}I$, where $I$ is a identity matrix with dimension $d$.
   \FOR{$j=0$ {\bfseries to} $n$}
   \STATE $\widetilde c_j  \leftarrow$ $j$-th coefficient of the Chebyshev interpolation of $f$ on $[a,b]$
   (see equation~\eqref{eq:ctildedef}).
   \ENDFOR
   \FOR {$i = 1$ { \bfseries to } $m$}
   \STATE Draw a random vector $\mathbf{v}^{(i)} \in \{ -1,+1\}^d$ whose entries are uniformly distributed\\
%   \STATE Draw a random vector $\mathbf{v}_i$ whose entries are either $\pm 1$ with equal probability.
   %  and $\mathbf{u} \leftarrow c_0 \ \mathbf{v}$
   
\iffalse
% new version
   \STATE $\mathbf{w}_0 \leftarrow \mathbf{v}_i$ and $\mathbf{w}_1 \leftarrow \widetilde A \mathbf{v}_i$
   \STATE $\mathbf{u} \leftarrow c_0 \mathbf{w}_0 + c_1 \mathbf{w}_1$
   \FOR {$j = 2$ { \bfseries to } $n$}
   	\STATE $\mathbf{w}_2 \leftarrow 2 \widetilde A \mathbf{w}_1 -  \mathbf{w}_0$
   	\STATE $\mathbf{u} \leftarrow \mathbf{u} + c_j \ \mathbf{w}_2$
   	\STATE $\mathbf{w}_0 \leftarrow \mathbf{w}_1$ and $\mathbf{w}_1 \leftarrow \mathbf{w}_2$
   \ENDFOR
\fi   
   
% old version %
   \STATE $\mathbf{w}_0^{(i)} \leftarrow \mathbf{v}^{(i)}$ and $\mathbf{w}_1^{(i)} \leftarrow \frac{2}{b-a}A\mathbf{v}^{(i)}-\frac{b+a}{b-a} \mathbf{v}^{(i)}$
   \STATE $\mathbf{u} \leftarrow \widetilde c_0 \mathbf{w}_0^{(i)} + \widetilde c_1 \mathbf{w}_1^{(i)}$
   \FOR {$j = 2$ { \bfseries to } $n$}
   \STATE $\mathbf{w}_2^{(i)} \leftarrow \frac{4}{b-a} A \mathbf{w}_1^{(i)} - \frac{2(b+a)}{b-a} \mathbf{w}_1^{(i)} - \mathbf{w}_0^{(i)}$
   \STATE $\mathbf{u} \leftarrow \mathbf{u} + \widetilde c_j \ \mathbf{w}_2$
   \STATE $\mathbf{w}_0^{(i)} \leftarrow \mathbf{w}_1^{(i)}$ and $\mathbf{w}_1^{(i)} \leftarrow \mathbf{w}_2^{(i)}$
   \ENDFOR
   \STATE $\Gamma \leftarrow \Gamma + \mathbf{v}^{(i)\top} \mathbf{u}/ m $

   \ENDFOR
   \STATE {\bfseries Output:} $\Gamma$
\end{algorithmic}
\end{algorithm}
%\textcolor{red}
{In Section \ref{sec:application}, we provide five concrete applications of the above algorithm:
approximating the log-determinant, the trace of matrix inverse, the Estrada index, the Schatten $p$-norm and testing positive definiteness, which correspond to $\log x$, $1/x$, $\exp(x)$, $x^{p/2}$ and $1-\mbox{sign}(x)$ respectively.
}

\subsection{Analysis}
We establish the following theoretical guarantee on the proposed algorithm.

\begin{theorem} \label{main}
Suppose function $f$ satisfies the followings:
%The assumptions we make about the scalar function $f$ are the followings:
%We first assume following conditions for function:
\begin{itemize}
\item $f$ is non-negative (or non-positive) on $[a,b]$.
%\item $f$ is analytic with $\left| f(z) \right| \le U$  for some $U<\infty$ on the elliptic region in the complex plane with foci at $a$ and $b$.
\item {$f$ is analytic with $\left| f\left(\frac{b-a}{2}z+\frac{b+a}{2}\right) \right| \le U$  for some $U<\infty$ on the elliptic region $E_\rho$ in the complex plane with foci at $-1,+1$ and $\rho$ as the sum of semi-major and semi-minor lengths.
}
%with $\left| f(z) \right| \le U$ in the region bounded by the ellipse with foci $\pm 1$ and major and minor semiaxis lengths summing to $\rho > 1$.
\item %$f(x)$ is always positive(or negative) and 
%$\left| {\min}_{x\in[a,b]} f(x) \right|\geq L$ for some $L>0$. % for $0<a<b$.
{$\min_{x\in[a,b]} \left| f(x) \right| \geq L$ for some $L>0$.}
%: \mathbb{R} \rightarrow \mathbb{R^-}$ is analytic in a neighborhood of $\left[ a,b\right] $ for $a,b > 0$
%\item $A \in \mathbb{R}^{d \times d}$ is positive semi-definite matrix with eigenvalues in $\left[ a,b\right] $
\end{itemize}
% with the conditions as mentioned earlier and $A \in \mathbb{R}^{d \times d}$ is a matrix 
%with eigenvalues in $[a,b]$ for $0<a<b$.
%\sout{
%We denote $\rho$ as the sum of semimajor and semiminor lengths of ellipse in the complex plane with foci at $-1$ and $+1$ where $ f\left(\frac{b-a}{2}z+\frac{b+a}{2}\right)$ is analytic on and inside.}

Given $\varepsilon, \zeta \in \left(0,1\right)$, if %the following lower bounds are satisfied that
\begin{align*}
m &\ge 54 \varepsilon^{-2} \log{\left( 2/\zeta \right)}, \\
n &\ge \log\left( \frac{8}{\varepsilon (\rho-1)}\frac{U}{L} \right)/\log \rho,
\end{align*}
then 
$$
\Pr \left(
\left| \Sigma_f(A) - \Gamma \right|
\le
\varepsilon \left| \Sigma_f(A)  \right|
\right)
\ge 1-\zeta.
$$
where $\Gamma$ is the output of {\bf Algorithm \ref{alg1}}.
\end{theorem}

The number of matrix-vector products performed by {\bf Algorithm \ref{alg1}} is $O(m n)$, thus the time-complexity  is $O(mn\|A\|_{\tt mv})$, where $\|A\|_{\tt mv}$ is that of the matrix-vector operation.
In particular, if $m,n=O(1)$, the complexity is linear with respect to $\|A\|_{\tt mv}$.
Therefore, Theorem \ref{main} implies that if $U,L=\Theta(1)$, then one can choose
$m,n=O(1)$ for $\varepsilon$-multiplicative approximation with probability of at least $1-\zeta$ given constants $\varepsilon,\zeta>0$.

%%%%%%%%%%%%%%%%%%
%%% Proof of Theorem 1 %%%
%%%%%%%%%%%%%%%%%%

\begin{proof}
%%% Choice of polynomial approximation degree
The condition    
$$
n \ge \log\left( \frac{8}{\varepsilon (\rho-1)}\frac{U}{L} \right)/\log \rho
$$
implies that 
\begin{align}
\frac{4U}{\left( \rho - 1\right) \rho^n} \le \frac{\varepsilon}{2} L\,.\label{choice}
\end{align}
%%% Bound between tr f(A) and tr pn(A)
Recall that the trace of a matrix is equal to the sum of its eigenvalues and this also holds for a function of the matrix, i.e., $f(A)$. 
Under this observation, we establish a matrix version of Corollary \ref{converge}. % for matrix version.
Let $\lambda_1, \dots, \lambda_d \in [a,b]$ be the eigenvalues of $A$. We have
\begin{align}
\left| \Sigma_f(A) - {\tt tr}\left( \widetilde p_n(A) \right) \right| 
&= \left| \sum_{i=1}^d f(\lambda_i) - \widetilde p_n\left( \lambda_i \right)\right| \nonumber
\le \sum_{i=1}^d \left| f(\lambda_i) - \widetilde p_n\left( \lambda_i \right)\right| \nonumber \\
&\le \sum_{i=1}^d \frac{4U}{\left( \rho - 1\right) \rho^n} = \frac{4dU}{\left( \rho - 1\right) \rho^n} \label{ieq1} \\
&\le \frac{\varepsilon}{2} dL \label{ieq2} \le \frac{\varepsilon}{2} d \min_{[a,b]} \left| f(x) \right| \\
&\le \frac{\varepsilon}{2} \sum_{i=1}^d \left| f(\lambda_i) \right| = \frac{\varepsilon}{2} \left| \Sigma_f(A) \right| \label{bound1}
\end{align}
where the inequality~\eqref{ieq1} is due to Corollary \ref{converge}, inequality \eqref{ieq2} holds due to inequality \eqref{choice}, and the last equality is due to the fact that $f$ is either non-negative or
non-positive.

Moreover, the inequality of \eqref{bound1} shows
\begin{align*}
\left| {\tt tr}\left( \widetilde p_n(A) \right) \right| - \left| \Sigma_f(A) \right| \le \left| \Sigma_f(A) - {\tt tr}\left( \widetilde p_n(A) \right)\right| \le \frac{\varepsilon}{2} \left| \Sigma_f(A) \right|
\end{align*}
which implies for $\varepsilon \in \left(0,1\right)$ that 
\begin{align}
\left| {\tt tr}\left( \widetilde p_n(A) \right) \right| \le \left(\frac{\varepsilon}{2} + 1 \right) \left| \Sigma_f(A) \right| \le \frac{3}{2} \left| \Sigma_f(A) \right|\,. \label{p2f}
\end{align}

A polynomial degree $n$ that satisfies \eqref{choice} also satisfies \eqref{condition},
and from this it follows that $\widetilde p_n(A)$ is positive semi-definite matrix by Lemma \ref{positive}. 
Hence, we can apply Theorem \ref{thm:trace}:  for $m \ge 54 \varepsilon^{-2} \log \left( 2/ \zeta \right)$
we have,
\begin{align*}
\Pr \left( 
\left| {\tt tr} \left( \widetilde p_n(A) \right) - {\tt tr}_m \left( \widetilde p_n(A) \right)\right| 
\le 
\frac{\varepsilon}{3} 
\left| {\tt tr} \left( \widetilde p_n(A) \right) \right| 
\right) 
\ge 
1-\zeta\,.
\end{align*}
In addition, this probability with \eqref{p2f} provides 
\begin{align}
\Pr \left( 
\left| {\tt tr} \left( \widetilde p_n(A) \right) - {\tt tr}_m \left( \widetilde p_n(A) \right)\right| 
\le 
\frac{\varepsilon}{2} 
\left| \Sigma_f(A) \right| 
\right) 
\ge 
1-\zeta. \label{bound2}
\end{align}

Combining \eqref{bound1} with \eqref{bound2} we have
\begin{align*}
1-\zeta 
&\le
\Pr \left( 
\left| {\tt tr} \left( \widetilde p_n(A) \right) - {\tt tr}_m \left( \widetilde p_n(A) \right)\right| 
\le 
\frac{\varepsilon}{2} 
\left| \Sigma_f(A) \right| 
\right) \\
&\le 
\Pr \Big( 
\left| \Sigma_f(A) - {\tt tr}\left( \widetilde p_n(A) \right) \right| + \left| {\tt tr}\left( \widetilde p_n(A) \right) - {\tt tr}_m\left( \widetilde p_n(A) \right) \right|\\
&\qquad \qquad \qquad \qquad  \qquad \qquad \le \frac{\varepsilon}{2} \left| \Sigma_f(A) \right| + \frac{\varepsilon}{2} \left| {\tt tr} \left( f(A) \right) \right|
\Big) \\
&\le 
\Pr \left(
\left| \Sigma_f(A) - {\tt tr}_m\left( \widetilde p_n(A) \right) \right| 
\le \varepsilon \left| \Sigma_f(A) \right|
\right)
\end{align*}
We complete the proof by observing that {\bf Algorithm \ref{alg1}} computes $\Gamma = {\tt tr}_m \left( \widetilde p_n(A) \right)$.
\end{proof}

%\subsection{Completion of Proof of Theorem \ref{thm:psd}} 
%\vspace{-0.1in}
\section{Applications}\label{sec:application}

In this section, we discuss several applications of {\bf Algorithm \ref{alg1}}: approximating the 
log-determinant, trace of the matrix inverse, the Estrada index, the Schatten $p$-norm
and testing positive definiteness. Underlying these applications is executing {\bf Algorithm \ref{alg1}} with  the following functions: $f(x) = \log x$ (for log-determinant), $f(x) = 1/x$ (for matrix inverse),
$f(x) =\exp(x)$ (for the Estrada index),
$f(x) = x^{p/2}$ (for the Schatten $p$-norm) and 
{
$f(x) = \frac12 \left( 1 + \tanh\left(-\alpha x\right)\right)$, as
a smooth approximation of $1 - \mbox{sign}(x)$ (for testing positive definiteness).
}

\subsection{Log-determinant of Positive Definite Matrices}
Since $\Sigma_{\log}(A) = \log \det A$ our algorithm can naturally be used to approximate the log-determinant.
However, it is beneficial to observe that 
$$
\Sigma_{\log}(A) = \Sigma_{\log}(A/(a+b)) + d \log(a+b)
$$
and use {\bf Algorithm~\ref{alg1}} to approximate $\Sigma_{\log}(\overline{A})$ for $\overline A = A / (a + b)$.
The reason we consider $\overline A$ instead of $A$ as an input of {\bf Algorithm \ref{alg1}}
is because  all eigenvalues of $\overline A$ are strictly less than 1 and the constant $L>0$ in Theorem \ref{main} is guaranteed to exist for $\overline A$. The procedure is summarized in the {\bf Algorithm \ref{alg2}}. In the 
next subsection we generalize the algorithm for general non-singular matrices.

% and then generalize to general non-singular matrix.

\begin{algorithm}[th!]
   \caption{Log-determinant approximation for positive definite matrices}
\begin{algorithmic}\label{alg2}
   \STATE {\bfseries Input:} positive definite matrix $A \in \mathbb{R}^{d \times d}$ with eigenvalues in $[a,b]$ for some $a,b>0$, sampling number $m$ and polynomial degree $n$\\
   \vspace{0.03in}
   \STATE {\bfseries Initialize:} $\overline{A} \leftarrow A/\left( a+b \right)$
   \vspace{0.03in}
   \STATE $\Gamma \leftarrow$ Output of {\bf Algorithm \ref{alg1}} with inputs $\overline{A}, \left[\frac{a}{a+b},\frac{b}{a+b}\right],m,n$ with $f(x)=\log x$
   \vspace{0.03in}
   \STATE $\Gamma \leftarrow \Gamma + d \log\left( a+b\right)$ 
   \vspace{0.03in}
   \STATE {\bfseries Output:} $\Gamma$ 
\end{algorithmic}
\end{algorithm}

We note that
{\bf Algorithm \ref{alg2}} requires us to know a positive lower bound $a>0$
%In most applications, the upper bound on eigenvalues is easy to obtain, e.g.,
%one can choose $$a = - \|A\|_{\infty} \qquad b= \|A\|_{\infty}.$$
%In some applications, we might need tighter bounds on $a,b$, which will be discussed later. 
%or one can run the power iteration \cite{ipsen1997computing} to estimate a better bound.
%On the other hand, the lower bound 
for the eigenvalues, which is in general
harder to obtain than the upper bound $b$ (e.g. one can choose $b=\|A\|_{\infty}$). % as we mentioned earlier.
In some special cases, the smallest eigenvalue of positive definite matrices are known, e.g.,
%not easy to obtain, except for special cases
%including 
random matrices \cite{tao2009inverse, tao2010random}
and diagonal-dominant matrices \cite{gershgorin1931uber, moravca2008bounds}.
%It is easy to obtain in the problem of counting spanning trees we studied in Section \ref{sec:spanning},
%and 
Furthermore, it is sometimes explicitly given as a parameter in many machine learning log-determinant applications
\cite{wainwright2006log}, e.g., $A= a I_d + B$ for some positive semi-definite matrix $B$
and this includes 
the application involving Gaussian Markov Random Fields (GMRF) in Section \ref{sec:gmrf}.
%In general, one can use the inverse power iteration \cite{ipsen1997computing} to estimate it.
%and assume that $a$ and $b$ are given.%, although it would be first necessary to estimate them.
%Although, in practice, this interval would be required to estimate first,
%we assume that $a$ and $b$ are given.
%However, in practice, 
%%% conditions for function f
%% Function setting
%We make assumption about the function $f$ as the followings:
%Under the assumptions, one can see that $f$ is always positive or negative in the interval $[a,b]$.

We provide the following theoretical bound on the sampling number $m$ and the polynomial degree $n$ of {\bf Algorithm \ref{alg2}}.

\begin{theorem}\label{thm:main1}
Given $\varepsilon, \zeta\in (0,1)$,
consider the following inputs for {\bf Algorithm \ref{alg2}}:
\begin{itemize}
\item $A \in \mathbb{R}^{d \times d}$ be a positive definite matrix with eigenvalues in $[a,b]$ for $a,b>0$
\item $ m \ge 54 \varepsilon^{-2} \left( \log\left(1+\frac{b}{a}\right)\right)^2\log{\left(\frac{2}{\zeta}\right)} $
\item $ n \ge \frac{\log{\left( \frac{20}{\varepsilon} \left( \sqrt{\frac{2b}{a}+1} - 1 \right) \frac{ \log\left(1+(b/a)\right) \log \left( 2 + 2 (b/a) \right)}{\log{\left( 1 + (a/b)\right)}} \right)}}{\log{\left( \frac{\sqrt{2 (b/a)+1}+1}{\sqrt{2 (b/a) + 1}-1}\right)}} 
= O\left( \sqrt{\frac{b}{a}} \log \left( \frac{b}{\varepsilon a} \right) \right)$
%$
\end{itemize}
Then, it follows that
%For given accuracy $\varepsilon_0$ and probability error $\zeta_0$, the \textit{estimated log-det}
{\begin{align*}
\Pr \left[ \ \left| \log \det A - \Gamma \right| \le  \varepsilon 
d
%\left( \left| \log \det A \right| + d \log\left( a+b\right) \right) 
\right] \ge 1 - \zeta
\end{align*}
}
where $\Gamma$ is the output of {\bf Algorithm \ref{alg2}}.
\end{theorem}
\begin{proof}
The proof of Theorem \ref{thm:main1} is straightforward using Theorem \ref{main} with choice of
upper bound $U$, lower bound $L$ and constant $\rho$ for the function $\log x$.
Denote $\delta = \frac{a}{a+b}$ and eigenvalues of $\overline{A}$ lie in the interval $[\delta,1-\delta]$. 
We choose the ellipse region, denoted by $E_\rho$, in the complex plane with foci at $+1, -1$ 
and its semi-major axis length is $1/(1-\delta)$.
Then, 
$$\rho = \frac1{1-\delta} + \sqrt{\left( \frac1{1-\delta}\right)^2-1}=\frac{\sqrt{2-\delta}+\sqrt{\delta}}{\sqrt{2-\delta}-\sqrt{\delta}}>1$$
and $\log\left( \frac{\left( 1-2\delta\right)x+1}{2}\right)$ is analytic on and inside $E_\rho$ in the complex plane. 

%The length of semiminor axis of the ellipse is equal to $\sqrt{\left( 1/(1-\delta) \right)^2 - 1}$. 
%Hence, the constant $\rho$ can be set to

The upper bound $U$ can be obtained as follows:
\begin{align*}
\max_{z \in E_\rho} \left| \log\left( \frac{\left( 1-2\delta\right)z+1}{2}\right) \right|
&\le \max_{z \in E_\rho} \sqrt{\left( \log \left| \frac{\left( 1-2\delta\right)z+1}{2} \right| \right)^2 + \pi ^2} \\
&= \sqrt{ \left( \log \left| \frac{\delta}{2\left( 1-\delta\right)} \right| \right)^2 + \pi ^2} \le 5 \log \left( \frac2{\delta} \right) := U.
\end{align*}
where the inequality in the first line holds because 
$\left| \log z \right| = \left| \log \left| z \right| + i \arg \left( z \right) \right| \le \sqrt{\left( \log \left| z\right| \right)^2 + \pi ^2}$ for any $z \in \mathbb{C}$
and equality in the second line holds by the maximum-modulus theorem.
We also have the lower bound on $\log x$ in $[\delta, 1-\delta]$ as follows:
\begin{align*}
\min_{[\delta,1-\delta]} \left| \log x\right| = \log\left(\frac1{1-\delta}\right) := L
\end{align*}
With these constants, a simple calculation reveals that Theorem~\ref{main} implies that {\bf Algorithm \ref{alg1}} approximates $\left| \log \det \overline{A} \right|$ with $\varepsilon/\log(1/\delta)$-mulitipicative approximation. 

The additive error bound now follows by using the  fact that
$
\left| \log \det \overline{A} \right| \le d \log \left( 1/\delta \right)\,.
$
\end{proof}

The bound on polynomial degree $n$ in the above theorem is relatively tight, e.g., 
$n = 27$ for $\delta=0.1$ and $\varepsilon=0.01$. 
% While our bound on sampling number $m$ is not tight,
% we observe that $m\approx 30$ is sufficient %enough 
% for high accuracy in our experiments.
% The time-complexity of {\bf Algorithm \ref{alg2}} is same as {\bf Algorithm \ref{alg1}} with $O(mn\|A\|_{\tt mv})$. 
% where $\|A\|_0$ is the number of non-zero entries of $A$.
Our bound for $m$ can yield very large numbers for the range of $\varepsilon$ and $\zeta$ we are interested
in. However, numerical experiments revealed that for the matrices we were interested in, the bound is not
tight and $m \approx 50$ was sufficient for the accuracy levels we required in the experiments.

\subsection{Log-determinant of Non-Singular Matrices} \label{sec:nonsingular}
One can apply the algorithm in the previous section to approximate
%our linear-time approximation scheme for 
the log-determinant of a {non-symmetric} non-singular
matrix $C \in \mathbb{R}^{d \times d}$.
The idea is simple: run {\bf Algorithm \ref{alg2}} on the positive definite matrix $C^{\top}C$.
The underlying observation is that
\begin{equation}
    \log|\det C| = \frac12 \log \det C^{\top} C\,.\label{eq:logdetC}
\end{equation}

%, through generalizing 
Without loss of generality, we assume that singular values of $C$ % of positive definite matrix $C^T C$ are
are in the interval $[\sigma_{\min}, \sigma_{\max}]$ for some $\sigma_{\min},\sigma_{\max}>0$, i.e.,
the condition number $\kappa(C)$ is at most $\kappa_{\max}: = \sigma_{\max}/\sigma_{\min}$.
The proposed algorithm is not sensitive to tight knowledge of $\sigma_{\min}$ or $\sigma_{\max}$, but
some loose lower and upper bounds on them, respectively, suffice.
%Given non-singular matrix $C$, one need to choose appropriate $\sigma_{\max},\sigma_{\min}$
%as mentioned before. %to run it. % the above algorithm.
A pseudo-code description appears as {\bf Algorithm \ref{alg3}}.

\begin{algorithm}[tbh!]
\caption{Log-determinant approximation for non-singular matrices}
\begin{algorithmic}\label{alg3}
\STATE {\bfseries Input:} non-singular matrix $C \in \mathbb{R}^{d \times d}$ with singular values 
are in the interval $[\sigma_{\min}, \sigma_{\max}]$ for some $\sigma_{\min},\sigma_{\max}>0$,
sampling number $m$ and polynomial degree $n$
\vspace{0.03in}
%\STATE {\bfseries Initialize:} $A \leftarrow C^TC$, $a \leftarrow \sigma_{\min}^2$, $b \leftarrow \sigma_{\max}^2$
%$\delta \leftarrow \frac{\sigma_{\min}^2}{\sigma_{\min}^2+\sigma_{\max}^2}$
\vspace{0.03in}
\STATE $\Gamma \leftarrow$ Output of {\bf Algorithm \ref{alg2}} for inputs $C^\top C, \left[ \sigma_{\min}^2,\sigma_{\max}^2\right], m, n$
\vspace{0.03in}
\STATE $\Gamma \leftarrow \Gamma / 2$
\vspace{0.05in}
\STATE {\bfseries Output:} $\Gamma$
\end{algorithmic}
\end{algorithm}

%{\bf Algorithm \ref{alg3}} %The above algorithm

The time-complexity of  {\bf Algorithm \ref{alg3}} is  $O(m n \|C\|_{\tt mv})=
O(m n \|C^{\top}C\|_{\tt mv})$ as well since {\bf Algorithm \ref{alg2}} requires the computation of a products of matrix $C^{\top}C$ and a vector, and that can be accomplished by first multiplying by $C$ and then by $C^\top$.  
We state the following additive error bound of the above algorithm.

\begin{corollary}\label{thm:main2}
Given $\varepsilon, \zeta\in (0,1)$,
consider the following inputs for {\bf Algorithm \ref{alg3}}:
\begin{itemize}
\item $C \in \mathbb{R}^{d \times d}$ be a matrix with singular values in $[\sigma_{\min}, \sigma_{\max}]$ for some $\sigma_{\min},\sigma_{\max}>0$
\item $m \ge \mathcal M \left(\varepsilon, \frac{\sigma_{\max}}{\sigma_{\min}}, \zeta \right)$ and $n \ge \mathcal N \left( \varepsilon, \frac{\sigma_{\max}}{\sigma_{\min}} \right)$, where
\end{itemize}
\vspace{-0.15in}
\begin{align*}
&\mathcal M(\varepsilon, \kappa, \zeta):=\frac{14}{\varepsilon^{2}} \left( \log \left( 1 + \kappa^2\right) \right)^2 \log{ \left( \frac2{\zeta}\right) }\\
&\mathcal N \left( \varepsilon, \kappa \right) := \frac{\log{\left( \frac{10}{\varepsilon} \left( \sqrt{2 \kappa^2 + 1}-1 \right) \frac{\log{ ( 2 + 2\kappa^2 )}}{\log(1+\kappa^{-2})} \right)}}{\log{\left( \frac{{\sqrt{2 \kappa^2 + 1}}+1}{{\sqrt{2 \kappa^2 + 1}}-1} \right)}} 
%&\qquad \quad \ 
= O \left( {\kappa} \log{ \frac{\kappa}{\varepsilon} } \right)
\end{align*}
Then, it follows that
\begin{align*}
\Pr \left[ \ \left| \log{\left( \left| \det C \right| \right)}- \Gamma \right| \le  \varepsilon d \ \right] \ge 1 - \zeta
\end{align*}
where $\Gamma$ is the output of {\bf Algorithm \ref{alg3}}.
\end{corollary}
\begin{proof}
Follows immediately from equation~\eqref{eq:logdetC} and Theorem~\ref{thm:main1}, and observing that
all the eigenvalues of $C^{\top} C$ are inside $[\sigma^2_{\min}, \sigma^2_{\max}]$.
\end{proof}

We remark that the condition number $\sigma_{\max}/\sigma_{\min}$ decides the complexity of {\bf Algorithm \ref{alg3}}. 
As one can expect, the approximation quality and algorithm complexity become worse as the condition
number increases, as polynomial approximation for $\log $ near
the point $0$ is challenging and requires higher polynomial degrees.

%%%%%%%%%%%%%%%%%%%
%%% [3] Trace of matrix inversion %%%
%%%%%%%%%%%%%%%%%%%
\subsection{Trace of Matrix Inverse} \label{sec:inverse}
In this section, we describe how to estimate the trace of matrix inverse.  Since this task amounts
to computing $\Sigma_f(A)$ for $f(x)=1/x$, we propose {\bf Algorithm~\ref{alg:inverse}}
which uses 
% {\bf Algorithm \ref{alg:inverse}} presenting approximation of the trace of matrix inversion 
{\bf Algorithm \ref{alg1}}
as a subroutine.

%%% Algorithm - Trace of matrix inversion %%%
\begin{algorithm}[th!]
   \caption{Trace of matrix inverse}
\begin{algorithmic}\label{alg:inverse}
   \STATE {\bfseries Input:} positive definite matrix $A \in \mathbb{R}^{d \times d}$ with eigenvalues in $[a,b]$ for some $a,b>0$, sampling number $m$ and polynomial degree $n$
   \vspace{0.03in}
   \STATE $\Gamma \leftarrow$ Output of {\bf Algorithm \ref{alg1}} for inputs $A,[a,b], m,n$ with $f(x)=\frac1{x}$.
   \vspace{0.03in}
   \STATE {\bfseries Output:} $\Gamma$ 
\end{algorithmic}
\end{algorithm}

We provide the following theoretical bounds on sampling number $m$ and polynomial degree $n$ of {\bf Algorithm \ref{alg:inverse}}.
%%% Theorem - Trace of matrix inversion %%%
\begin{theorem}\label{thm:inverse}
Given $\varepsilon, \zeta\in (0,1)$,
consider the following inputs for {\bf Algorithm \ref{alg:inverse}}:
\begin{itemize}
\item $A \in \R^{d \times d}$ be a positive definite matrix with eigenvalues in $[a,b]$ %for some $a,b>0$
\item $ m \ge 54 \varepsilon^{-2} \log{\left(\frac{2}{\zeta}\right)} $
\item $ n \ge \log\left( \frac8{\varepsilon} \left( \sqrt{2\left( \frac{b}{a} \right)-1}-1 \right) \frac{b}{a} \right) 
/ \log \left( \frac2{\sqrt{2\left( \frac{b}{a} \right)-1}-1} + 1\right)
= O\left( \sqrt{\frac{b}{a}} \log \left( \frac{b}{\varepsilon a} \right)\right)
$
\end{itemize}
Then, it follows that
%For given accuracy $\varepsilon_0$ and probability error $\zeta_0$, the \textit{estimated log-det}
{\begin{align*}
\Pr \left[ \ \left| {\tt tr}\left( A^{-1} \right) - \Gamma \right| \le  \varepsilon 
\left| {\tt tr}\left( A^{-1}\right) \right|
\right] \ge 1 - \zeta
\end{align*}
}
where $\Gamma$ is the output of {\bf Algorithm \ref{alg:inverse}}.
\end{theorem}

%%% Theorem Proof - Trace of matrix inversion %%%
\begin{proof}
In order to apply Theorem \ref{main},
we define inverse function with linear transformation $\widetilde f$ as
$$\widetilde f \left(x\right) = \frac1{ \frac{b-a}{2}x + \frac{b+a}{2}} \mbox{\quad for \ } x \in [-1,1].$$
Avoiding singularities of $\widetilde f$,
it is analytic on and inside elliptic region in the complex plane passing through $\frac{b}{b-a}$
whose foci are $+1$ and $-1$.
The sum of length of semi-major and semi-minor axes is equal to
$$
\rho = \frac{b}{b-a} + \sqrt{\frac{b^2}{\left(b-a\right)^2} - 1} = 
\frac{2}{\sqrt{2\left( \frac{b}{a} \right) - 1} - 1} + 1.
$$

For the maximum absolute value on this region,
$\widetilde f$ has maximum value $U = 2/a$ at $- \frac{b}{b-a}$.
The lower bound is $L = 1/b$.
Putting those together, Theorem \ref{main}, 
implies the bounds stated in the theorem statement.
\end{proof}

%%%%%%%%%%%%%%%%%%%
%%% [4] Estrada Index %%%
%%%%%%%%%%%%%%%%%%%
\subsection{Estrada Index} \label{sec:estrada}

Given a (undirected) graph $G=(V,E)$, the Estrada index ${\rm EE}\left( G \right)$ is defined as
$${\rm EE}\left( G \right) := \Sigma_{\exp}(A_G) = \sum_{i=1}^d \exp(\lambda_i),$$
where $A_G$ is the adjacency matrix of $G$ and $\lambda_1, \dots, \lambda_{|V|}$ are the eigenvalues
of $A_G$. It is a well known result in spectral graph theory that the eigenvalues of $A_G$ are contained
in $[-\Delta_G, \Delta_G]$ where $\Delta_G$ is maximum degree of a vertex in $G$. Thus, the Estrada index $G$ can be computed using {\bf Algorithm \ref{alg1}} with the choice of $f(x)=\exp(x)$, $a=-\Delta_G$, and $b=\Delta_G$.
However, we state our algorithm and theoretical bounds in terms of a general interval $[a,b]$ that bounds
the eigenvalues of $A_G$, to allow for an a-priori tighter bounds on the eigenvalues (note, however, that it is
well known that always $\lambda_{\max} \geq \sqrt{\Delta_G}$).

%%% Algorithm - Estrada Index %%%
\begin{algorithm}[th!]
   \caption{Estrada index approximation}
\begin{algorithmic}\label{alg:estrada}
   \STATE {\bfseries Input:} adjacency matrix $A_G \in \mathbb{R}^{d \times d}$ with eigenvalues in $[a,b]$, sampling number $m$ and polynomial degree $n$\\
   \COMMENT{If $\Delta_G$ is the maximum degree of $G$, then $a=-\Delta_G, b=\Delta_G$ can be used as default.} 
   \vspace{0.03in}
   \STATE $\Gamma \leftarrow$ Output of {\bf Algorithm \ref{alg1}} for inputs $A,[a,b], m,n$ with $f(x)=\exp(x)$.
   \vspace{0.03in}
   \STATE {\bfseries Output:} $\Gamma$ 
\end{algorithmic}
\end{algorithm}

We provide the following theoretical bounds on sampling number $m$ and polynomial degree $n$ of {\bf Algorithm \ref{alg:estrada}}.

%%% Theorem - Estrada Index %%%
\begin{theorem}\label{thm:estrada}
%Let $G=(V,E)$ be a graph with adjacency matrix $A \in \mathbb{R}^{\left| V(G) \right| \times \left| V(G) \right|}$
%and ${\rm EE}\left( G \right)$ be Estrada index of $G$.
Given $\varepsilon, \zeta\in (0,1)$,
consider the following inputs for {\bf Algorithm \ref{alg:estrada}}:
\begin{itemize}
\item $A_G \in \R^{d \times d}$ be an adjacency matrix of a graph with eigenvalues in $[a,b]$.
\item $ m \ge 54 \varepsilon^{-2} \log{\left(\frac{2}{\zeta}\right)} $
\item $ n \ge \log \left( \frac{2}{\pi \varepsilon} (b-a) \exp\left( \frac{\sqrt{16\pi^2 + (b-a)^2} + (b-a)}{2} \right) \right)
/ \log \left( \frac{4\pi}{b-a}  + 1 \right)
= O\left( 
\frac{
b-a  +\log  \frac1{\varepsilon} 
}{
\log \left( \frac1{b-a}\right)
}
\right)
$
\end{itemize}
Then, it follows that
%For given accuracy $\varepsilon_0$ and probability error $\zeta_0$, the \textit{estimated log-det}
{\begin{align*}
\Pr \left[ \ \left| {\rm EE}\left( G\right) - \Gamma \right| \le  \varepsilon 
\left| {\rm EE}\left( G\right) \right|
\right] \ge 1 - \zeta
\end{align*}
}
where $\Gamma$ is the output of {\bf Algorithm \ref{alg:estrada}}.
\end{theorem}

%%% Theorem Proof - Estrada Index %%%
\begin{proof}
%The proof is similar as Theorem \ref{thm:main1}.
We consider exponential function with linear transformation as
$$\widetilde f \left(x\right) = \exp \left( \frac{b-a}{2}x + \frac{b+a}{2}\right) \mbox{\quad for \ } x \in [-1,1].$$
The function $\widetilde f$ is analytic on and inside elliptic region in the complex plane 
which has foci $\pm 1$ and passes through $\frac{4\pi i}{\left(b-a\right)}$.
The sum of length of semi-major and semi-minor axes becomes
$$
\frac{4\pi}{b-a} + \sqrt{\frac{16\pi^2}{\left(b-a\right)^2} + 1}
$$
and we may choose $\rho$ as $\frac{4\pi}{({b-a})} + 1$.

By the maximum-modulus theorem, the absolute value of $\widetilde f$ on this elliptic region is maximized at 
$\sqrt{\frac{16\pi^2}{\left(b-a\right)^2} + 1}$
with value $U = \exp\left( \frac{\sqrt{16 \pi^2 + \left( b-a\right)^2} +  (b+a)}{2} \right)$
and the lower bound has the value $L = \exp\left( a\right)$. 
Putting those all together in Theorem \ref{main}, 
we could obtain above the bound for approximation polynomial degree.
This completes the proof of Theorem \ref{thm:estrada}.
%$L=\exp(a)$ and $U=\exp\left( \sqrt{\frac{16\pi^2}{\left( b-a\right)^2}+1} \right)$ by the maximum-modulus theorem.
\end{proof}

%%%%%%%%%%%%%%%%%%%
%%% [5] Schatten p-norm %%%
%%%%%%%%%%%%%%%%%%%
\subsection{Schatten $p$-Norm} \label{sec:schatten}
The Schatten $p$-norm for $p \geq 1$ of a matrix $M \in \R^{d_1\times d_2}$ is defined as
$${\|M \|_{(p)}} = \left( \sum_{i=1}^{\min\{d_1,d_2 \}} \sigma_i^p\right)^{1/p}
$$
where $\sigma_i$ is the $i$-th singular value of $M$ for $1 \leq i \leq \min\{d_1,d_2 \}$.
Schatten $p$-norm is widely used in linear algebric applications such as
nuclear norm (also known as the trace norm) for $p=1$:
$${\|M \|_{(1)}} = {\tt tr}\left( \sqrt{M^\top M} \right) = \sum_{i=1}^{\min\{d_1,d_2 \}} \sigma_i.
$$
The Schatten $p$-norm corresponds to 
the spectral function $x^{p/2}$ of matrix $M^\top M$ since
singular values of $M$ are square roots of eigenvalues of $M^\top M$.
%thus we can approximate it using the proposed algorithm.
In this section, we assume that general (possibly, non-symmetric) non-singular
matrix $M \in \R^{d_1 \times d_2}$
has singular values in the interval $[\sigma_{\min}, \sigma_{\max}]$ for some $\sigma_{\min},\sigma_{\max}>0$, and
propose {\bf Algorithm \ref{alg:schatten}} which uses
%{\bf Algorithm \ref{alg:schatten}} describes approximation of Schatten $p$-norm as a subroutine of {\bf Algorithm \ref{alg1}}.
{\bf Algorithm \ref{alg1}} as a subroutine.

%%% Algorithm - Schatten p-norm %%%
\begin{algorithm}[tbh!] 
   \caption{Schatten $p$-norm approximation}
\begin{algorithmic}\label{alg:schatten}
   \STATE {\bfseries Input:} matrix $M \in \R^{d_1 \times d_2}$ with singular values in $\left[ \sigma_{\min}, \sigma_{\max}\right]$, sampling number $m$ and polynomial degree $n$\\
   \vspace{0.03in}
   \STATE $\Gamma \leftarrow$ Output of {\bf Algorithm \ref{alg1}} for inputs $M^\top M,\left[ \sigma_{\min}^2, \sigma_{\max}^2\right], m,n$ with $f(x)=x^{p/2}$.
   \vspace{0.03in}
   \STATE $\Gamma \leftarrow \Gamma^{1/p}$ 
   \STATE {\bfseries Output:} $\Gamma$ 
\end{algorithmic}
\end{algorithm}

%
%\newpage
We provide the following theoretical bounds on sampling number $m$ and polynomial degree $n$ of {\bf Algorithm \ref{alg:schatten}}.

%%% Theorem - Schatten p-norm %%%
\begin{theorem}\label{thm:schatten}
Given $\varepsilon, \zeta\in (0,1)$,
consider the following inputs for {\bf Algorithm \ref{alg:schatten}}:
\begin{itemize}
\item $M \in \R^{d_1 \times d_2}$ be a matrix with singular values in $\left[ \sigma_{\min}, \sigma_{\max}\right]$
\item $ m \geq 54 \varepsilon^{-2} \log{\left(\frac{2}{\zeta}\right)} $
\item $ n \geq \mathcal N \left( \varepsilon, p, \frac{\sigma_{\max}}{\sigma_{\min}} \right)$, where 
\begin{align*}
\mathcal N \left( \varepsilon, p, \kappa \right) := \log\left( \frac{16 \left( \kappa - 1\right)}{\varepsilon} \left( \kappa^2 + 1\right)^{p/2} \right) / \log \left( \frac{\kappa+1}{\kappa-1} \right) = O\left( \kappa \left( p \log \kappa + \log \frac1{\varepsilon} \right) \right).
\end{align*}
\end{itemize}
Then, it follows that
%For given accuracy $\varepsilon_0$ and probability error $\zeta_0$, the \textit{estimated log-det}
%\textcolor{red}
{\begin{align*}
\Pr \left[ \ \left| \|M\|_{(p)}^p - \Gamma^p \right| \le  \varepsilon 
\|M\|_{(p)}^p 
\right] \ge 1 - \zeta
\end{align*}
}
where $\Gamma$ is the output of {\bf Algorithm \ref{alg:schatten}}.
\end{theorem}

%%% Theorem Proof - Schatten p-norm %%%
\begin{proof}
%The proof is similar as Theorem \ref{thm:main1}.
%The proof of Theorem \ref{thm:main1} is straightforward using Theorem \ref{main} with choice of upper bound $U$, lower bound $L$ and constant $\rho$ for the function $\log x$.
Consider following function as 
$$\widetilde f \left(x\right) =  \left( \frac{\sigma_{\max}^2-\sigma_{\min}^2}{2}x + \frac{\sigma_{\max}^2+\sigma_{\min}^2}{2}\right)^{p/2} \mbox{\quad for \ } x \in [-1,1].$$
In general, $x^{p/2}$ for arbitrary $p \geq 1$ is defined on $x \geq 0$.
We choose elliptic region $E_\rho$ in the complex plane 
such that it is passing through 
$- \left( \sigma_{\max}^2+\sigma_{\min}^2 \right) / \left( \sigma_{\max}^2-\sigma_{\min}^2 \right)$
and having foci $+1, -1$ on real axis
so that $\widetilde f$ is analytic on and inside $E_\rho$.
The length of semi-axes can be computed as
$$
\rho
= \frac{\sigma_{\max}^2+\sigma_{\min}^2}{\sigma_{\max}^2-\sigma_{\min}^2}
+ \sqrt{\left( \frac{\sigma_{\max}^2+\sigma_{\min}^2}{\sigma_{\max}^2-\sigma_{\min}^2}\right)^2 - 1} 
%\frac{\sigma_{\max}^2+\sigma_{\min}^2}{\sigma_{\max}^2-\sigma_{\min}^2} + \sqrt{\frac{\sigma_{\max}^2+\sigma_{\min}^2}{\sigma_{\max}^2-\sigma_{\min}^2} + 1}
%= \frac{\kappa^2 + 1}{\kappa^2 - 1} + \sqrt{\left( \frac{\kappa^2 + 1}{\kappa^2 - 1}\right)^2 - 1} 
= \frac{\sigma_{\max}+\sigma_{\min}}{\sigma_{\max}-\sigma_{\min}}
= \frac{\kappa_{\max}+1}{\kappa_{\max}-1}
$$
where $\kappa_{\max} = \sigma_{\max} / \sigma_{\min}$.
%Then, $\widetilde f $ is analytic on and inside $E_\rho$ in the complex plane.

The maximum absolute value is occurring at 
$\left( \sigma_{\max}^2+\sigma_{\min}^2 \right) / \left( \sigma_{\max}^2-\sigma_{\min}^2 \right)$
and its value is $U = \left( \sigma_{\max}^2+\sigma_{\min}^2 \right) ^{p/2}$.
Also, the lower bound is obtained as $L = {\sigma_{\min}}^{p}$.
Applying Theorem $\ref{main}$ together with choices of $\rho$, $U$ and $L$, 
the bound of degree for polynomial approximation $n$ can be achieved. 
This completes the proof of Theorem \ref{thm:schatten}.
\end{proof}

%%%%%%%%%%%%%%%%%%%
%%% [5] Testing PSD %%%
%%%%%%%%%%%%%%%%%%%
\subsection{Testing Positive Definiteness} \label{sec:psd}

In this section we consider the problem of determining if a given symmetric matrix $A \in \R^{d \times d}$ 
is positive definite. This can be useful in several scenarios. For example, when solving a linear system
$Ax=b$, determination if $A$ is positive definite can drive algorithmic choices like whether to use Cholesky decomposition or use $LU$ decomposition, or alternatively, if an iterative method is preferred, whether to use 
CG or MINRES. In another example, checking if the Hessian is positive or negative definite can help
determine if a critical point is a local maximum/minimum or a saddle point.

In general, positive-definiteness can be tested in $O(d^3)$ operations by attempting
a Cholesky decomposition of the matrix. If the operation succeeds then the 
matrix is positive definite, and if it fails (i.e., a negative diagonal is encountered) the matrix
is indefinite. If the matrix is sparse, running time can be improved as long
as the fill-in during the sparse Cholesky factorization is not too big, but 
in general the worst case is still $\Theta(d^3)$. More in line with this paper is 
to consider the matrix implicit, that is accessible only via matrix-vector products.
In this case, one can reduce the matrix to tridiagonal form by doing $n$ iterations of Lanczos,
and then test positive definiteness of the reduced matrix. This requires 
$d$ matrix vector multiplications, so running time $\Theta(\|A\|_{\tt mv} \cdot d)$.
However, we note that this algorithm is not a
practical algorithm since it suffers from severe numerical instability.

In this paper we consider testing positive definiteness under the property testing framework. 
Property testing algorithms relax the requirements of decision problems by allowing them 
to issue arbitrary answers for inputs that are on the boundary of the class. That is,
for decision problem on a class $L$ (in this case, the set of positive definite matrices)
the algorithm is required to accept $x$ with high probability if $x \in L$, and reject
$x$ if $x \not\in L$ and $x$ is $\varepsilon$-far from any $y \in L$. For $x$'s that are
not in $L$ but are less than $\varepsilon$ far away, the algorithm is free to return
any answer. We say that such $x$'s are in the {\em indifference region}.
In this section we show that testing positive definiteness in the property testing
framework can be accomplished using $o(d)$ matrix-vector products. 

Using the spectral  norm of a matrix to measure distance, this suggests the following 
property testing variant of determining if a matrix is positive definite.

\begin{problem}
\label{problem:psd_full}
Given a symmetric matrix $A \in \R^{d \times d}$, $\varepsilon > 0$ and $\zeta \in (0, 1)$
\begin{itemize}
\item If $A$ is positive definite, accept the input with probability of at least $1-\zeta$.
\item If $\lambda_{\min} \leq -\varepsilon \| A \|_2$, reject the input with probability of at least $1-\zeta$.
\end{itemize}
\end{problem}

For ease of presentation, it will be more convenient to restrict the norm of $A$ to be at most
$1$, and for the indifference region to be symmetric around $0$.
\begin{problem}
\label{problem:psd_simpler}
Given a symmetric $A \in \R^{n\times n}$ with $\|A\|_2 \leq 1$, $\varepsilon > 0$ and $\zeta \in (0, 1)$
\begin{itemize}
\item If $\lambda_{\min} \geq \varepsilon/2$, accept the input with probability of at least $1-\zeta$.
\item If $\lambda_{\min} \leq -\varepsilon/2$, reject the input with probability of at least $1-\zeta$.
\end{itemize}
\end{problem}
It is quite easy to translate an instance of Problem~\ref{problem:psd_full} to an instance of 
Problem~\ref{problem:psd_simpler}. First we use power-iteration to approximate $\|A\|_2$. Specifically,
we use enough power iterations to guarantee that with a normally distributed random initial vector 
to find a $\lambda'$ such that $|\lambda' - \|A\|_2 | \leq \left( \varepsilon / 2 \right) \|A\|_2 $
with probability of at least $1-\zeta/2$.
Due to a bound by Klien and Lu~\cite[Section 4.4]{klein1996power} we need to perform 
\begin{equation*}
\left\lceil \frac{2}{\varepsilon}\left(\log^2\left(2d\right)+\log\left(\frac{8}{\varepsilon\zeta^{2}}\right)\right)\right\rceil
\end{equation*}
iterations (matrix-vector products) to find such an $\lambda'$. Let $\lambda = \lambda' / (1 - \varepsilon / 2)$ and consider 
\begin{equation*}
    B = \frac{A - \frac{\lambda \varepsilon}{2} I_d}{(1 + \frac{\varepsilon}{2})\lambda}\,.
\end{equation*}
It is easy to verify that $\|B\|_2 \leq 1$ and 
%that if $\varepsilon$ is small enough so that 
$\lambda / \|A\|_2 \geq 1/2$ for $\varepsilon > 0$.
If $\lambda_{\min}(A) \in [0, \varepsilon \|A\|_2]$ then $\lambda_{\min}(B) \in [-\varepsilon'/2, \varepsilon' / 2]$ where $\varepsilon' = \varepsilon / (1 + \varepsilon/2)$. Therefore, by solving Problem~\ref{problem:psd_simpler} on $B$ with $\varepsilon'$ and $\zeta' = \zeta / 2$ we have a solution to 
Problem~\ref{problem:psd_full} with $\varepsilon$ and $\zeta$.

We call the region $[-1, -\varepsilon/2] \cup [\varepsilon/2, 1]$ the {\em active region} 
$\cal{A}_\varepsilon$, and the interval $[-\varepsilon/2, \varepsilon/2]$ as the {\em indifference region} 
$\cal{I}_\varepsilon$.
%Allowing an indifference region $\cal{I}_\varepsilon$, 
%we can construct 
%a polynomial approximation of $S(\cdot)$ well enough to 
%an algorithm for testing the definiteness of $A$
%as long as the smallest eigenvalue of $A$ is not in the indifference region.
%In particular, we show that if we have a polynomial that is an $\eta$-approximation over the 
%active region, and is bounded below by $1/8n$ on the entire $[-1,1]$ interval,
%then that polynomial can be used to test positive definiteness.

Let $S$ be the reverse-step function, that is,
\begin{align*}
S \left( x \right) \ = 
\begin{cases}
1      &  \mbox{if }  x \leq 0, \\
0      &  \mbox{if }  x > 0.
\end{cases}
\end{align*}
Now note that a matrix $A \in \R^{d \times d}$ is positive definite if and only if
\begin{equation}
 \Sigma_S(A) \leq \gamma \label{eqn:exact_eig_test}
\end{equation}
for any fixed $\gamma \in(0, 1)$.
This already suggests using {\bf Algorithm \ref{alg1}} to test positive definite,
however the discontinuity of $S$ at $0$ poses problems. 

To circumvent this issue we use a two-stage approximation. 
First, we approximate the reverse-step function 
using a smooth function $f$  (based on the hyperbolic tangent), 
and then use {\bf Algorithm \ref{alg1}} to approximate
$\Sigma_f(A)$. By carefully controlling the transition in $f$, the degree
in the polynomial approximation and the quality of the trace estimation, 
we guarantee that as long as the smallest eigenvalue is not in the 
indifference region, the {\bf Algorithm \ref{alg1}} will return 
less than 1/4 with high probability if $A$ is positive definite and will return more than 
1/4 with high probability if $A$ is not positive definite.
The procedure is summarized as  {\bf Algorithm \ref{alg:psd}}. 

\begin{algorithm}[tbh!] 
   \caption{Testing positive definiteness}
\begin{algorithmic}\label{alg:psd}
   \STATE {\bfseries Input:} symmetric matrix $A \in \R^{d \times d}$ with eigenvalues in $\left[ -1, 1\right]$, sampling number $m$ and polynomial degree $n$\\
   \STATE Choose $\varepsilon >0$ as the distance of active region
   \STATE $\Gamma \leftarrow$ Output of {\bf Algorithm \ref{alg1}} for inputs $A,\left[ -1, 1\right], m,n$ with $f(x)= \frac1{2}\left(1+\tanh(-\frac{\log\left( 16d\right)}{\varepsilon}x) \right)$.
   \IF{$\Gamma < \frac{1}{4}$}
      \STATE {\bf return} PD %POSITIVE DEFINITE
   \ELSE
      \STATE {\bf return} NOT PD%POSITIVE DEFINITE
   \ENDIF
\end{algorithmic}
\end{algorithm}

The correctness of the algorithm is established in the following theorem. While we use
{\bf Algorithm~\ref{alg1}}, the indifference region requires a more careful analysis so the proof
does not rely on Theorem~\ref{main}.
%%% Theorem - Test PSD %%%
\begin{theorem}\label{thm:psd}
Given $\varepsilon, \zeta\in (0,1)$,
consider the following inputs for {\bf Algorithm \ref{alg:psd}}:
\begin{itemize}
\item $A \in \R^{d \times d}$ be a symmetric matrix with eigenvalues in $\left[ -1, 1\right]$
and $\lambda_{\min}(A)\not\in {\cal I}_\varepsilon$ where $\lambda_{\min}(A)$ is the minimum eigenvalue 
of $A$.
\item $ m \geq 24 \log{\left(\frac{2}{\zeta}\right)} $
\item $ n \geq  \frac{\log\left( 32 \sqrt{2} \log\left( 16d\right)\right)  + \log\left( 1/\varepsilon\right) - \log\left( \pi/8d \right)}{\log\left( 1 + \frac{\pi \varepsilon}{4 \log\left( 16d\right)}\right)} 
= O\left(\frac{\log^2(d) + \log(d)\log(1/\varepsilon)}{\varepsilon}  \right)$
\end{itemize}
Then the answer returned by {\bf Algorithm \ref{alg:psd}} is correct with probability of at least $1-\zeta$.
\end{theorem}

The number of matrix-vector products in {\bf Algorithm \ref{alg:psd}} is 
$O\left(\left(\frac{\log^2(d) + \log(d)\log(1/\varepsilon)}{\varepsilon}  \right)\log(1/\zeta)\right)$
as compared with $O(d)$ that are required with non-property testing previous methods.

\begin{proof}
Let $p_n$ be the degree Chebyshev interpolation of $f$. We begin by showing that 
$$
\max_{x \in \cal{A}_\varepsilon}\left|S(x) - p_n(x)\right| \leq \frac{1}{8d}\,.
$$
To see this, we first observe that 
$$
\max_{x \in \cal{A}_\varepsilon}\left|S(x) - p_n(x)\right| \leq \max_{x \in \cal{A}_\varepsilon}\left|S(x) - f(x)\right| + \max_{x \in \cal{A}_\varepsilon}\left|f(x) - p_n(x)\right|
$$
so it is enough to bound each term by $1/16d$.

For the first term, let 
\begin{equation}
\alpha = \frac1{\varepsilon} \log \left( 16d\right) \label{eq:alpha}
\end{equation}
and note that $f(x)= \frac{1}{2} ( 1 + \tanh(- \alpha x))$.
We have 
%$\left| f(x) - S(x) \right| \leq \xi/2$ for $x \in \cal{A}_\varepsilon$.
\begin{align*}
\max_{x \in \cal{A}_\varepsilon}\left|S(x) - f(x)\right| 
&= \frac1{2} \ \max_{x \in [\varepsilon/2, 1]}\left|1 - \tanh(\alpha x)\right| \\
&= \frac1{2} \left(1 - \tanh\left( \frac{\alpha \varepsilon}{2}\right) \right) \\
&= \frac{e^{-\alpha  \varepsilon}}{1 + e^{-\alpha  \varepsilon}} \\
&\leq e^{-\alpha  \varepsilon} = \frac{1}{16d}.
\end{align*} 

To bound the second term we use Theorem~\ref{converge}. 
To that end we need to define an appropriate ellipse. Let $E_\rho$ 
be the ellipse with  foci $-1,+1$ passing through $\frac{i\pi}{4\alpha}$. 
The sum of semi-major and semi-minor axes is equal to
$$
\rho = \frac{\pi + \sqrt{\pi^2 + 16\alpha^2}}{4\alpha}.
$$
The poles of $\tanh$ are of the form $i\pi/2\pm ik\pi$ so $f$ is analytic 
inside $E_\rho$. It is always the case that $|\tanh(z)| \leq 1$ if $\Im(z) \leq \pi / 4$
\footnote{To see this, note that using simple algebraic manipulations it is possible to show that
$|\tanh(z)| = (e^{2\Re(z)} + e^{2\Re(z)} - 2\cos(2\Im(z)))/ (e^{2\Re(z)} + e^{2\Re(z)} - 2\cos(2\Im(z)))$,
from which the bound easily follows.},
so $|f(z)| \leq 1$ for $z \in E_\rho$.
Applying Theorem~\ref{converge} and noticing that $\rho \geq 1 + \pi / 4\alpha$, we have
\begin{align*}
\max_{x \in [-1,1]}\left| p_n(x) - f(x) \right| 
& \leq \frac{4}{(\rho - 1)\rho^{d}}  \leq \frac{16\alpha}{\pi(1 + \frac{\pi}{4\alpha})^d}.
\end{align*}
Thus, $\max_{x \in [-1,1]}\left| p_n(x) - f(x) \right| \leq \frac{1}{16d}$ provided that
$$
n \geq \frac{\log(32 \alpha) - \log(\pi/8d)}{\log(1 + \frac{\pi}{4\alpha})}\,.
$$
which is exactly the lower bound on $n$ in the theorem statement.

Let
$$B = p_n \left(A\right) + \frac{1}{8d} I_d$$ 
%where $I_d$ is the $d$-dimension identity matrix.
then
$B$ is symmetric positive semi-definite since $p_n(x) \geq -1/8d$
due to the fact that $|f(x)| \geq 0$ for every $x$. 
According to Theorem~\ref{thm:trace},
$$
\Pr\left( \left| {\tt tr}_m \left( B \right) - {\tt tr}\left(B\right)  \right| \leq \frac{{\tt tr}\left(B\right)}{2}\right) \geq 1 - \zeta
$$
if $m \geq 24 \log\left( 2/\zeta \right)$ as assumed in the theorem statement.

Since ${\tt tr}_m \left(B\right) = {\tt tr}_m \left(p_n(A) \right) + 1/8$,
${\tt tr} \left(B\right) = {\tt tr} \left(p_n(A)\right) + 1/8$, and 
$\Gamma = {\tt tr}_m \left(p_n(A) \right)$ we have 
\begin{equation}
    \label{eq:gamma_pn}
\Pr \left( \left| \Gamma - {\tt tr} \left(p_n(A)\right) \right| \leq \frac{{\tt tr} \left(p_n(A)\right)}{2} + 1/16 \right) \geq 1 - \zeta\,.
\end{equation}

If $\lambda_{\min}(A) \geq \varepsilon / 2$, then all eigenvalues of $S(A)$ are zero and so
all eigenvalues of $p_n(A)$ are bounded by $1/8d$, so ${\tt tr} \left(p_n(A)\right) \leq 1/8$.
Inequality~\eqref{eq:gamma_pn} then imply that 
$$
\Pr\left( \Gamma \leq 1/4 \right) \geq 1 - \zeta\,. 
$$

If $\lambda_{\min}(A) \leq -\varepsilon / 2$, $S(A)$ has at least one eigenvalue that is 
$1$ and is mapped in $p_n(A)$ to at least $1 - 1/8d \geq 7/8$. All other eigenvalues in $p_n(A)$
are at the very least $-1/8d$ so ${\tt tr} \left(p_n(A)\right) \geq 3/4$. Inequality~\eqref{eq:gamma_pn} then imply that  
$$
\Pr\left(\Gamma \geq 1/4 \right) \geq 1 - \zeta\,. 
$$

The conditions $\lambda_{\min}(A) \geq \varepsilon / 2$ and  $\lambda_{\min}(A) \leq -\varepsilon / 2$ together cover all cases for $\lambda_{\min}(A) \not \in {\cal I}_\varepsilon$ thereby completing the proof.
\end{proof}

\section{Experiments} \label{sec:exp}
The experiments were performed using a machine with 3.5GHz Intel i7-5930K processor with 12 cores and 32 GB RAM.
We choose $m=50$, $n=25$ in our algorithm
%and iterate $100$ times to measure the approximation error
unless stated otherwise.

\begin{figure}[t!] 
\begin{center} 
%\vskip -0.08in
\subfigure[]{\includegraphics[width = 0.33\textwidth]{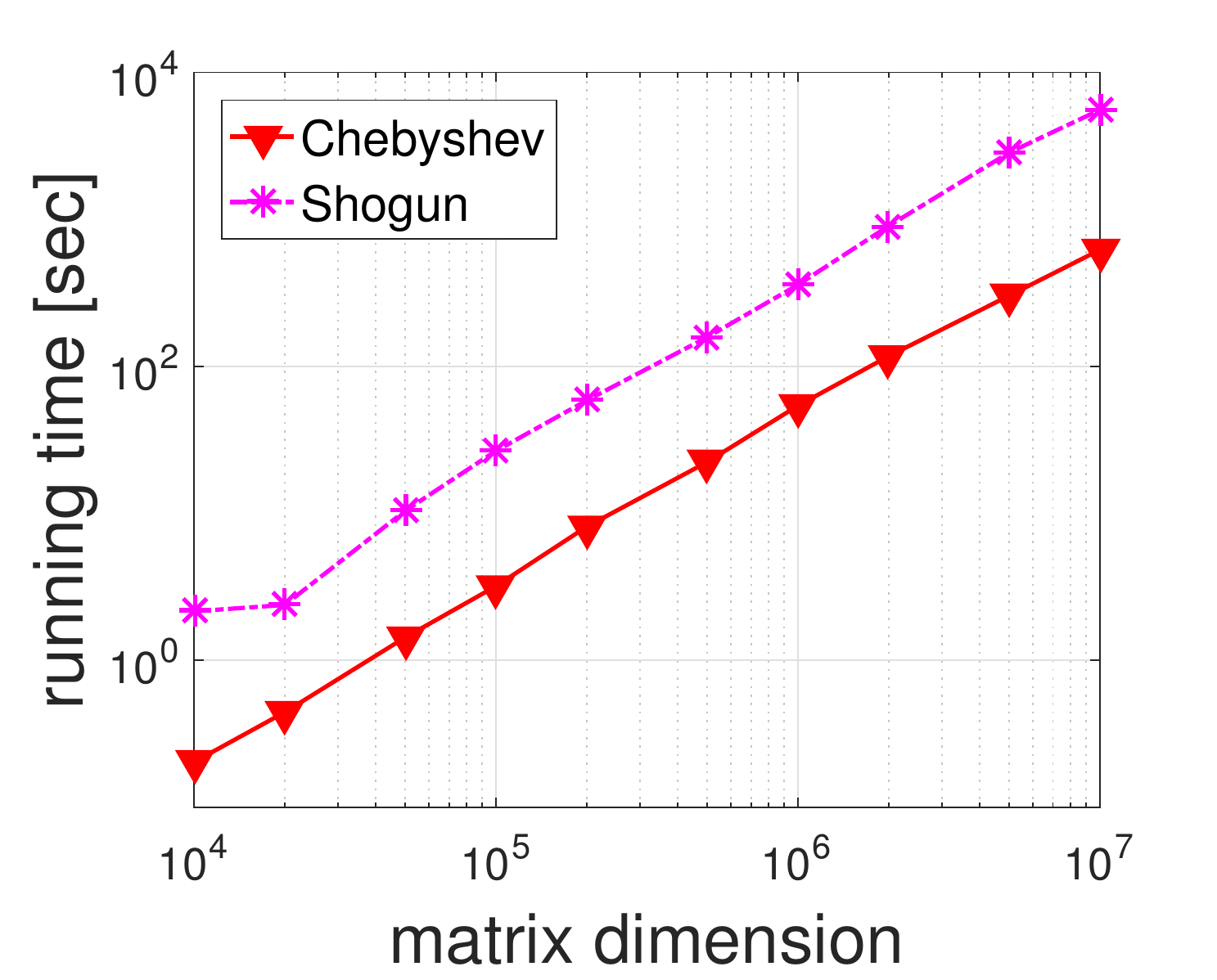}} 
\hskip -0.01in
\subfigure[]{\includegraphics[width = 0.33\textwidth]{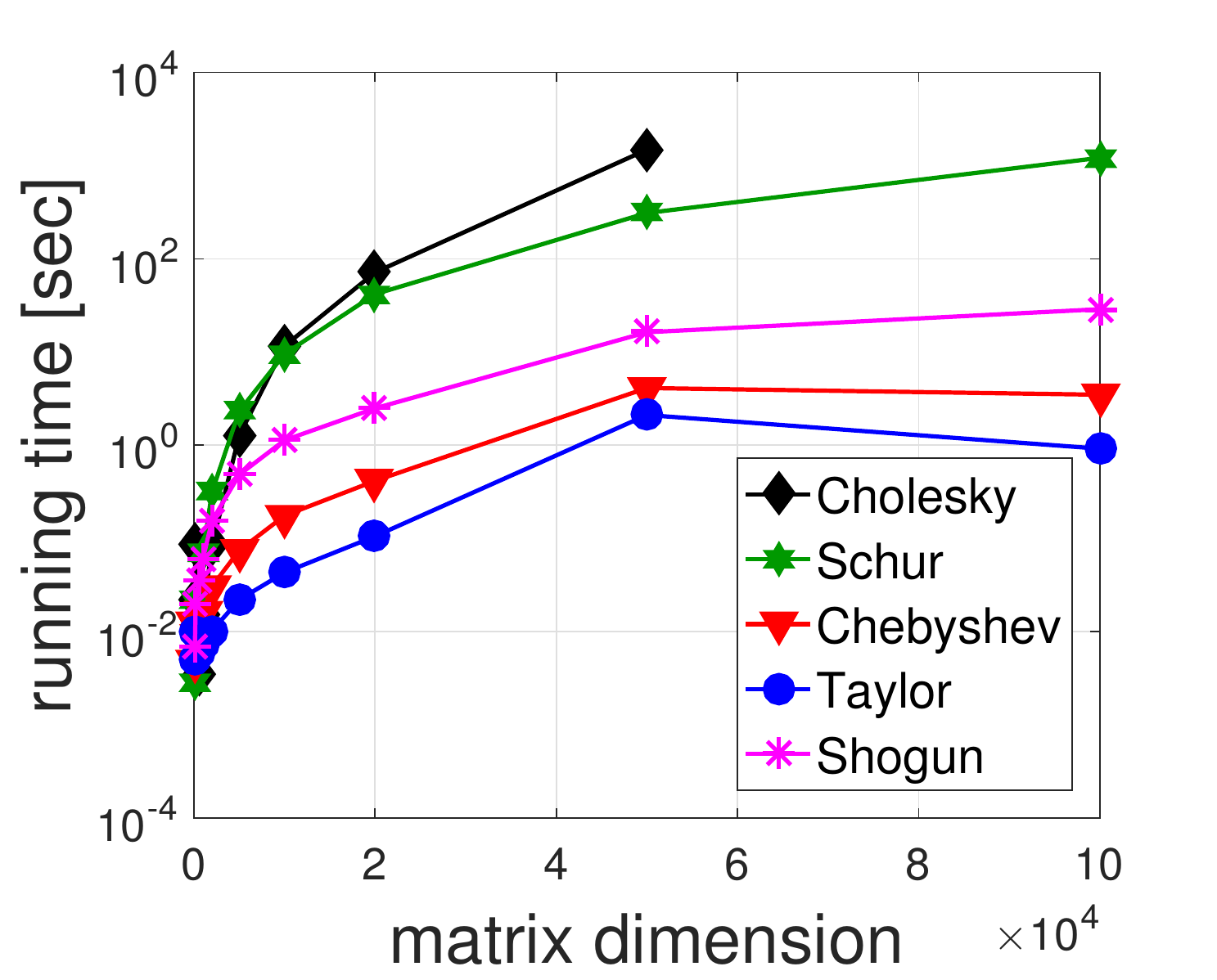}}
\hskip -0.01in
\subfigure[]{\includegraphics[width = 0.33\textwidth]{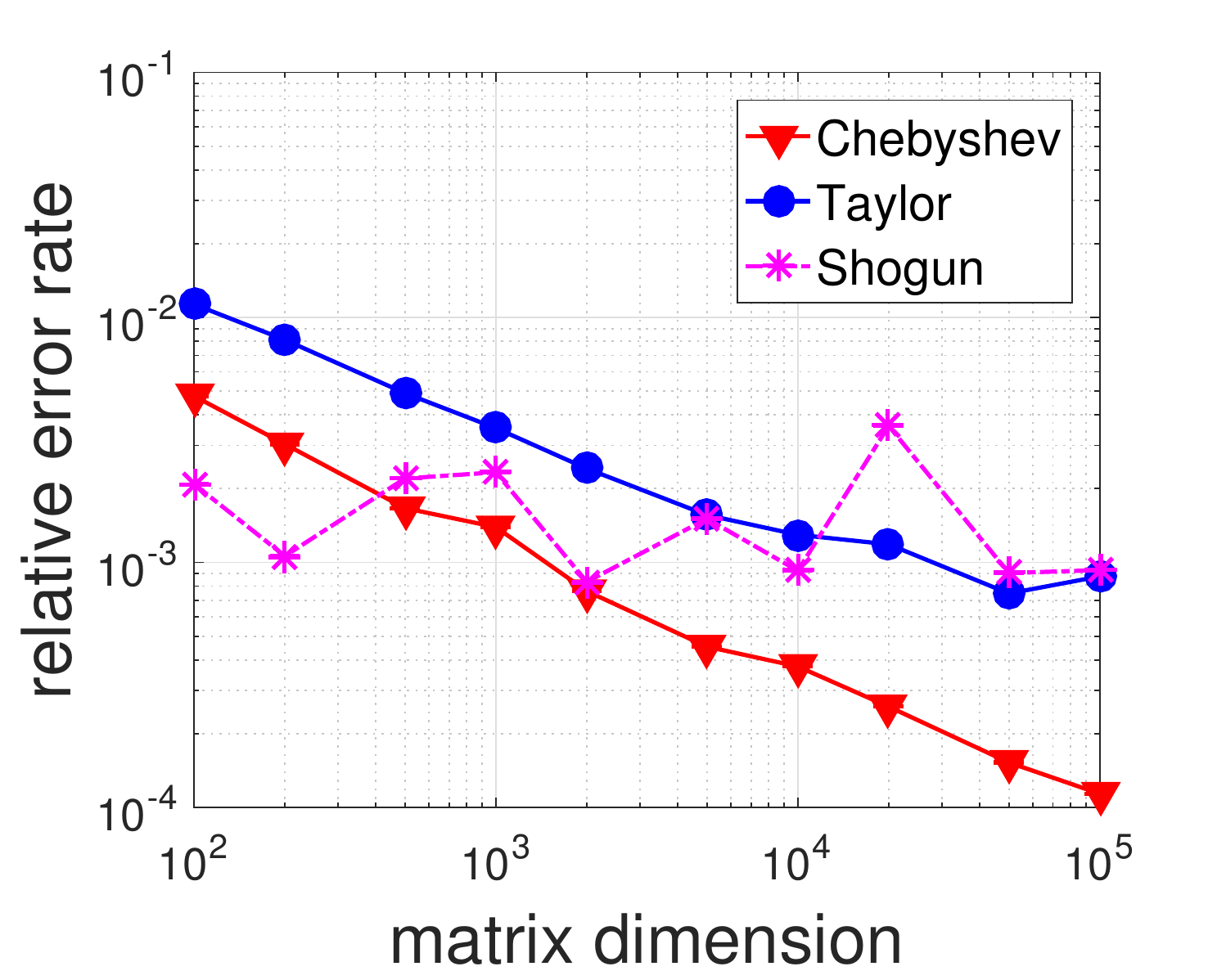}}
%\vskip -0.16in
\subfigure[]{\includegraphics[width = 0.33\textwidth]{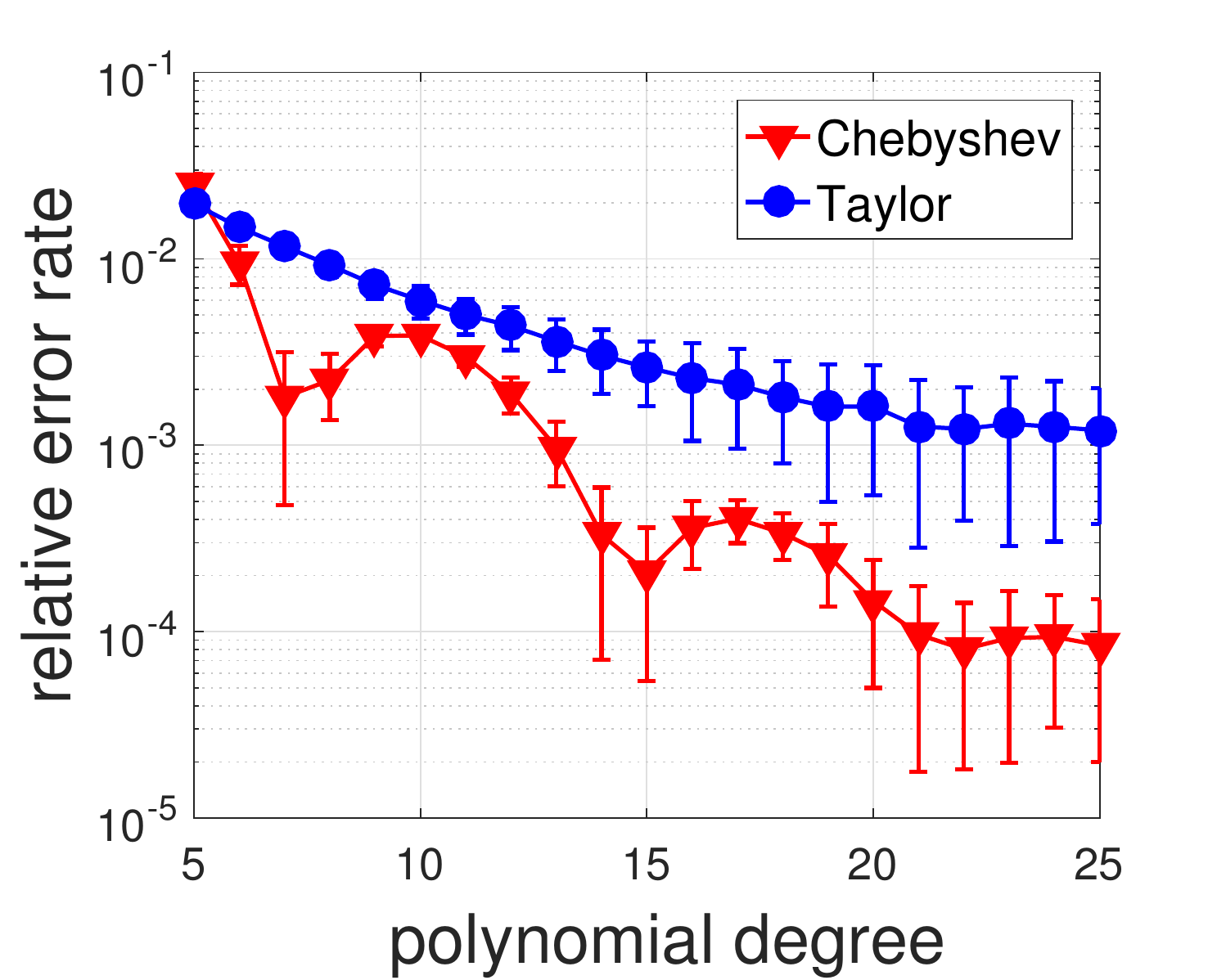}}
\hskip -0.01in
\subfigure[]{\includegraphics[width = 0.33\textwidth]{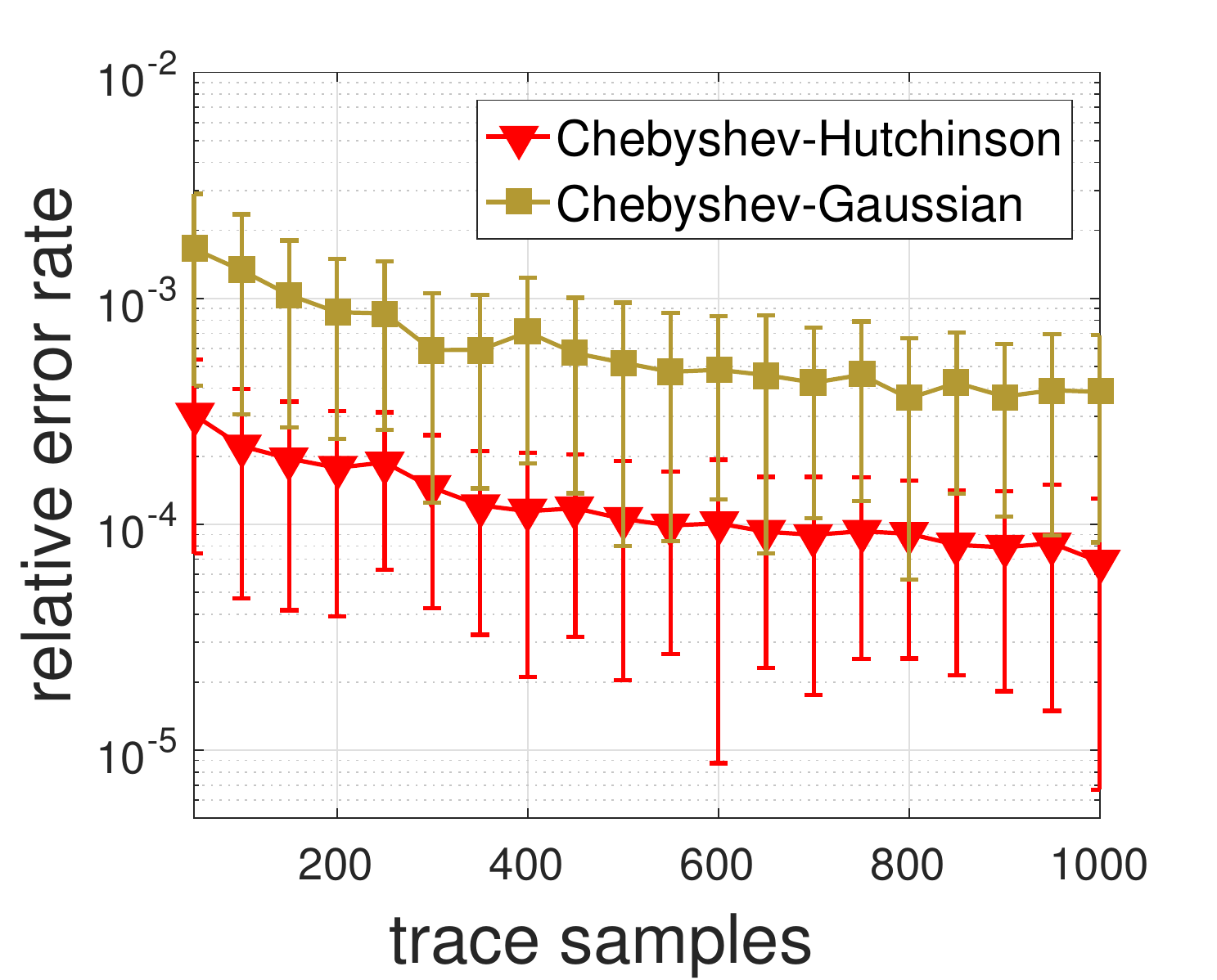}}
\hskip -0.01in
\subfigure[]{\includegraphics[width = 0.33\textwidth]{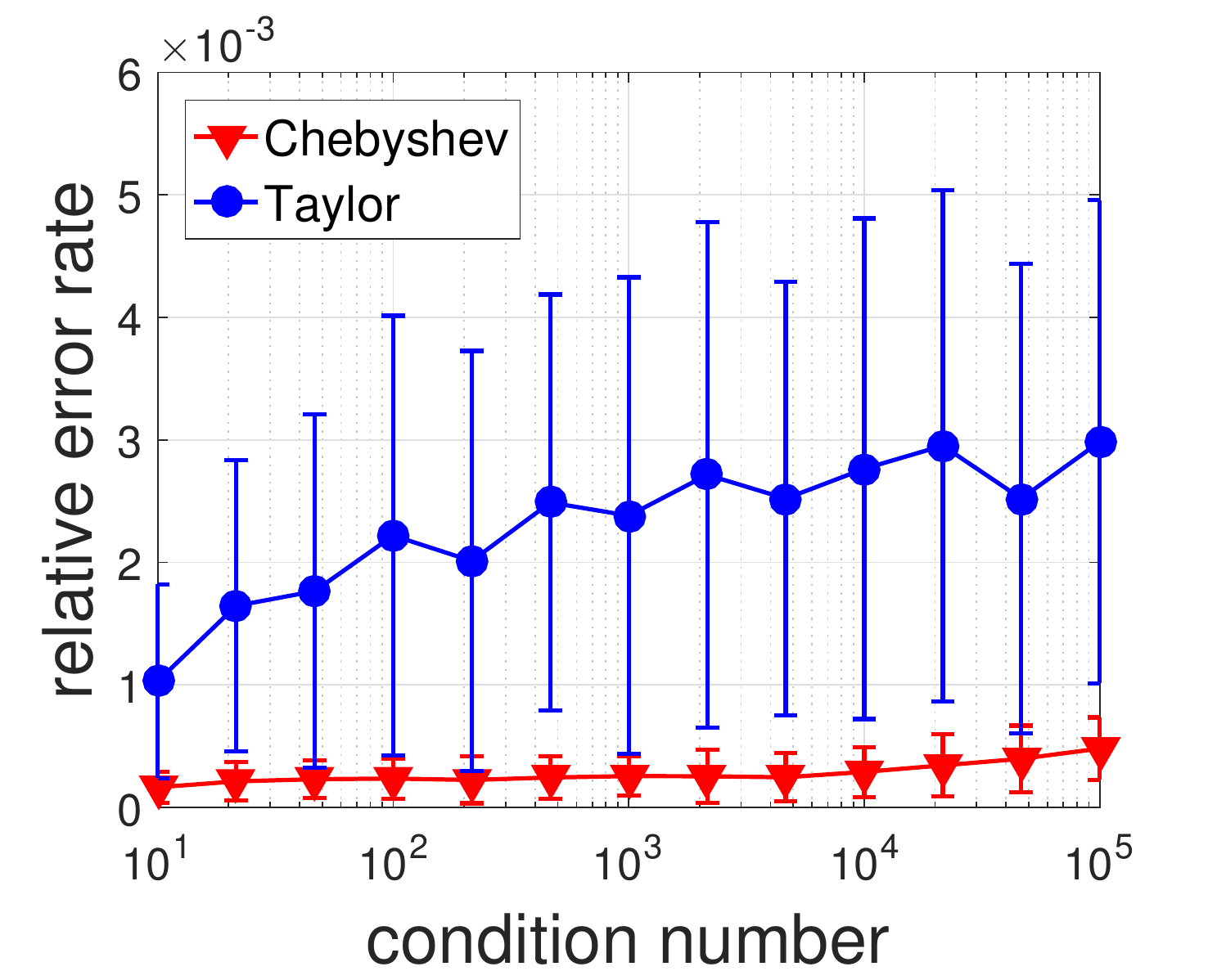}}
%\vspace{-0.1in}
\caption{Performance evaluations of {\bf Algorithm \ref{alg2}} (i.e., Chebyshev) and comparisons with other algorithms:
(a) running time varying matrix dimension,  (b) comparison in running time among Cholesky decomposition, Schur complement \cite{hsieh2013big}, Cauchy integral formula \cite{aune2014_GP} and Taylor-based algorithm \cite{leithead_GP_stoch},
%\textcolor{red}{
%(c) relative error of those approximation schemes vs. \ dimension.
%and (d) comparison in accuracy with Taylor-based algorithm \cite{leithead_GP_stoch}
%with respect to polynomial degree
(c) relative error varying matrix dimension, 
(d) relative error varying polynomial degree,
(e) relative error varying the number of trace samples,
(f) relative error varying condition number.
%comparing with Taylor-based algorithm
%for random matrix $A\in \mathbb{R}^{d\times d}$ with $d=10^5$
%, including Cholesky decomposition,
%Shur completion and Talor-based algorithm  \cite{leithead_GP_stoch}.
The relative error means a ratio between the absolute error of
the output of an  approximation algorithm and the actual value of log-determinant.}\label{fig:performance}
\end{center}
\end{figure}

%\subsection{Performance Evaluation and Comparion}
\subsection{Log-determinant}
In this section, we report the performance of our algorithm compared to other methods for computing
the log-determinant of positive definite matrices.
We first investigate the empirical performance of the proposed algorithm on large sparse random matrices.
We generate a random matrix $A\in \mathbb{R}^{d\times d}$, where the number of non-zero entries per each row is around $10$.
We first select non-zero off-diagonal entries in each row with values drawn from the standard normal distribution.
To make the matrix symmetric, we set the entries in transposed positions to the same values.
Finally, to guarantee positive definiteness, we set its diagonal entries to absolute row-sums and add a small margin value 0.1.
Thus, the lower bound for eigenvalues can be chosen as $a=0.1$ and the upper bound is set to the infinite norm of a matrix.
%which is equal to the smallest eigenvalue of the matrix. 

Figure \ref{fig:performance} (a) shows the running time of 
%log-determinant approximation (
{\bf Algorithm \ref{alg2}}
from matrix dimension $d=10^4$ to $10^7$. %, where we choose $m=50$, $n=25$.
%,$\sigma_{\min}=10^{-3}$ and $\sigma_{\max}=\|C\|_{1}$. 
The algorithm scales roughly linearly over a large range of matrix sizes as expected. 
%\textcolor{red}{
In particular, it takes only $600$ seconds for a matrix of dimension $10^7$ with $10^8$ non-zero entries. 
Under the same setup, we also compare the running time of our algorithm
with other ones including Cholesky decomposition and Schur
complement. The latter was used for sparse inverse covariance estimation with
over a million variables \cite{hsieh2013big} and we run the code implemented by the authors.
The running time of the algorithms are reported in Figure \ref{fig:performance} (b). 
Our algorithm is dramatically faster than both exact methods. 
Moreover, our algorithm is an order of magnitude
faster than the recent approach based on Cauchy integral formula~\cite{aune2014_GP},
%\footnote{http://www.shogun-toolbox.org/},
while it achieves better accuracy as reported in Figure \ref{fig:performance} (c).\footnote{The method \cite{aune2014_GP} is implemented in the SHOGUN machine learning toolbox, http://www.shogun-toolbox.org.}

%We use a machine with 3.40 Ghz Intel I7 processor with $24$ GB RAM. 
We also compare the relative accuracies between our algorithm and
that using Taylor expansions \cite{leithead_GP_stoch} {with the same sampling number} $m=50$ and polynomial degree $n=25$, as reported in Figure \ref{fig:performance} (c). 
We see that the Chebyshev interpolation based method more accurate than the one based on Taylor approximations. 
To complete the picture, we also use a large number of samples for trace estimator, $m=1000$, for both algorithms
to focus on the polynomial approximation errors. The results are reported in Figure \ref{fig:performance} (d), 
showing that our algorithm using Chebyshev expansions is superior
in accuracy compared to the Taylor-based algorithm.

In Figure \ref{fig:performance} (e),
%\textcolor{red}{
%For the issue of trace estimators, 
we compare two different trace estimators, Gaussian and Hutchinson,
under the choice of %where we choose the 
polynomial degree $n=100$.
We see that the Hutchinson estimator outperforms the Gaussian estimator.
%the result that Gaussian trace estimator is less accurate that Hutchinson's one.
Finally, in Figure \ref{fig:performance} (f) we report the results of experiments
 with varying condition number.
We see that the Taylor-based method is more sensitive to the condition number than the Chebyshev-based
method.
%which can expected from the convergence rate.
%}

\iffalse
We also compare the accuracy of our algorithm to a related stochastic algorithm 
that uses Taylor expansions \cite{leithead_GP_stoch}.
For a fair comparison we use a large number of samples for trace estimator, $m=1000$, for both algorithms
to focus on the polynomial approximation errors. The results are reported in Figure
\ref{fig:performance} (d), showing that our algorithm using Chebyshev expansions is superior
in accuracy compared to the one based on Taylor series.
\fi

\begin{figure*}[]
\begin{center} 
\subfigure[]{\includegraphics[width = 0.4\textwidth]{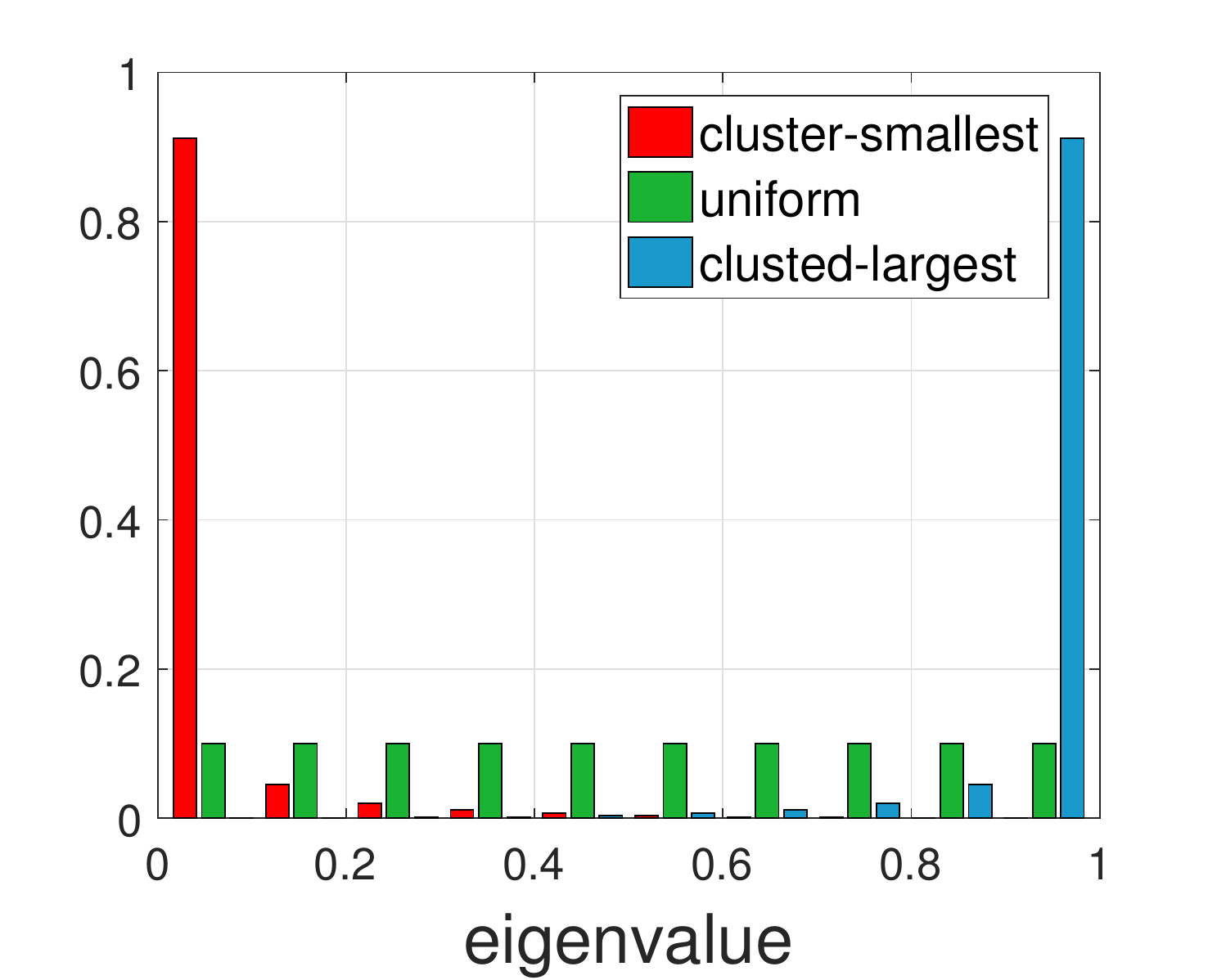}}
\subfigure[]{\includegraphics[width = 0.4\textwidth]{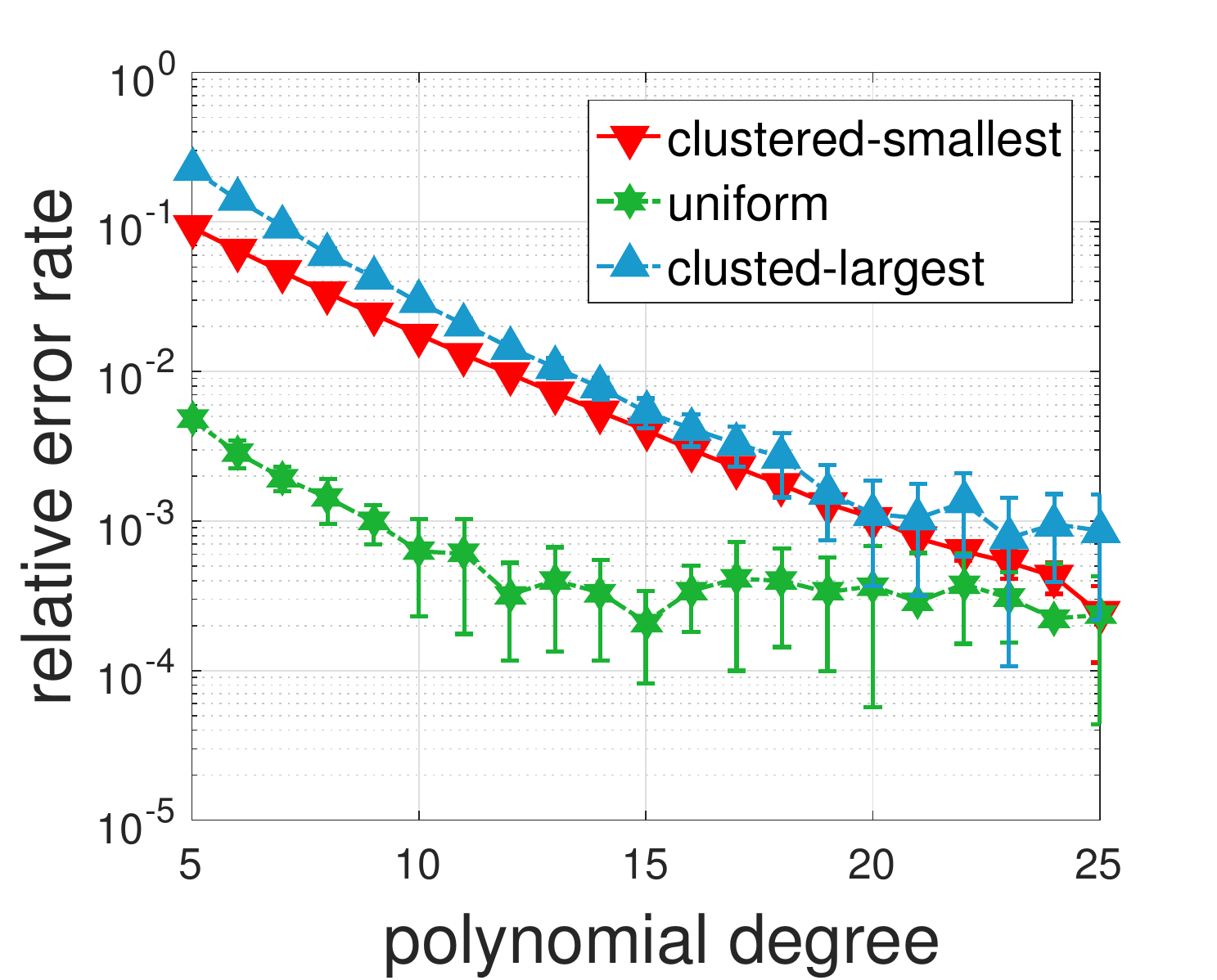}}
\caption{
%\textcolor{red}{
Performance evaluations of {\bf Algorithm \ref{alg2}}
%Comparison log-determinant approximation 
when eigenvalue distributions are uniform (green), clustered on the smallest one (red)
and clustered on the largest one (cyan):
(a) distribution of eigenvalues,
(b) relative error varying polynomial degree.
} \label{fig:eigenhist}
\end{center}
\end{figure*}

%\textcolor{red}{
Chebyshev expansions have extreme points more likely 
around the end points of the approximating interval since
the absolute values of their derivatives are larger. 
Hence, one can expect that %It turns out that 
if eigenvalues are clustered on the smallest (or largest) one, 
the quality of approximation becomes worse.
To see this, 
we run {\bf Algorithm \ref{alg2}} for matrices having
%matrix with 
uniformly distributed eigenvalues and
%matrix whose 
eigenvalues clustered on the smallest (or largest) one, %, respectively.
which is reported in Figure \ref{fig:eigenhist}.
%(a) shows
%distribution of eigenvalue for each case. 
We observe that if the polynomial degree is small, 
the clustering effect cause larger errors, but 
the error decaying rate with respect to polynomial degree
is not sensitive to it.

%it can be
%alleviated using large polynomial degree ($\geq 25$).

%Although matrix with uniformly distributed eigenvalues has higher accuracy,
%the convergence rate, which is not depends on eigenvalue distribution, 
%are similar for all matrices as shown in Figure \ref{fig:eigenhist} (b).
%}

\subsection{Maximum Likelihood Estimation for GMRF using Log-determinant} \label{sec:gmrf}

In this section, we apply our proposed algorithm approximating log determinants
for maximum likelihood (ML) estimation in Gaussian Markov Random Fields (GMRF) \cite{rue_GMRF}.
GMRF is a multivariate joint Gaussian distribution defined with respect to a graph. 
Each node of the graph corresponds to a random variable in the Gaussian distribution, where the graph captures the
conditional independence relationships (Markov properties) among the random variables. The model has been extensively used in many applications in computer vision, spatial statistics, and other fields.
The inverse covariance matrix $J$ (also called information or precision matrix) is positive definite
and sparse: $J_{ij}$ is non-zero only if the edge $\{i,j\}$ is contained in the graph.
We are specifically interested in the problem of parameter estimation from data 
(fully or partially observed samples from the GMRF), 
where we would like to find the maximum likelihood estimates of 
the non-zero entries of the information matrix. 
%(Otherwise the transition to details of GMRF is abrupt).
%}

\noindent{\bf GMRF with 100 million variables on synthetic data.}
We first consider a GMRF on a square grid of size $5000\times 5000$ with precision matrix 
$J \in \mathbb{R}^{d \times d}$ with $d=25\times 10^6$, 
which is parameterized by $\eta$, i.e., each node has
four neighbors with partial correlation $\eta$. We generate a sample $\mathbf x$ from the GMRF model
(using Gibbs sampler) for parameter $\eta=-0.22$. The log-likelihood of the sample is
%\textcolor{red}
% $$\log p({\mathbf x}| \eta) = \log \det J(\eta) - {\mathbf x}^\top J(\eta) {\mathbf x} + \mbox{"some terms independent of $\eta$"},$$
%\textcolor{red}{
$$\log p({\mathbf x}| \eta) = \frac12 \log \det J(\eta) - \frac12 {\mathbf x}^\top J(\eta) {\mathbf x} - \frac{d}{2} \log\left( 2 \pi\right),$$
%}
where $ J(\eta)$ is a matrix of dimension $25\times 10^6$ and $10^8$ non-zero entries.
%, and $G$ is a
%constant independent of $\eta$. 
Hence, the ML estimation requires to solve 
$$\max_{\eta} \left(\frac12 \log \det J(\eta) - \frac12 {\mathbf x}^\top J(\eta) {\mathbf x} - \frac{d}{2} \log\left( 2 \pi\right)\right).$$
We use {\bf Algorithm \ref{alg2}} to estimate the log-likelihood
as a function of $\eta$, as reported in Figure \ref{fig:rho}.
This confirms that the estimated log-likelihood is maximized at
the correct (hidden) value ${\eta} = -0.22$.

\begin{figure*}[tbh]
\begin{center}
%\vskip -0.15in
\includegraphics[width=0.48\textwidth]{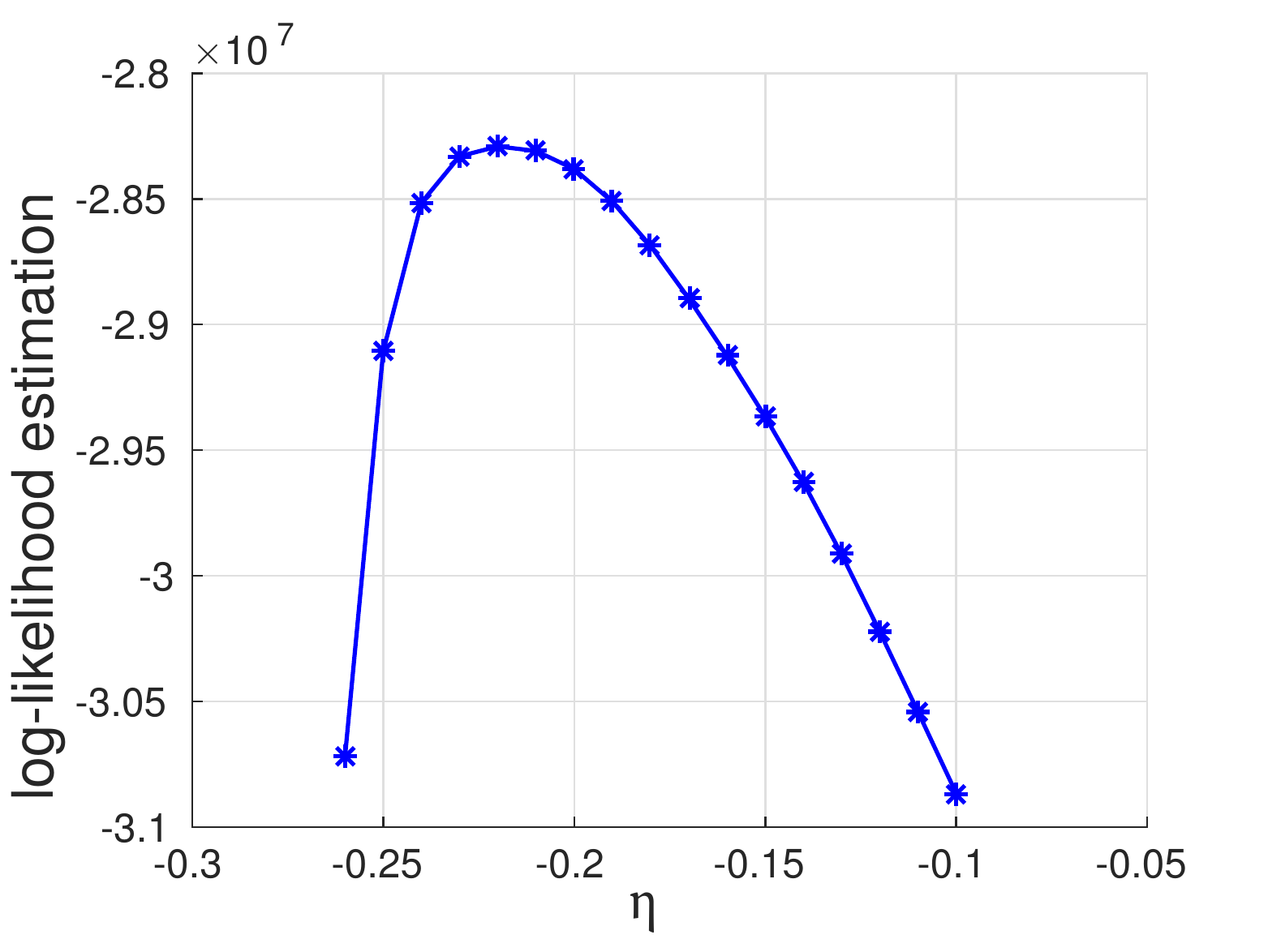}
\vskip -0.15in
\caption{Log-likelihood estimation for hidden parameter $\eta$ for square GMRF model of size $5000\times 5000$.}
\label{fig:rho}
\end{center}
%\vspace{-0.15in}
\end{figure*}

\noindent{\bf GMRF with 6 million variables for ozone data.}
We also consider a similar GMRF parameter estimation from real spatial data with missing values.
We use the data-set from \cite{aune2014_GP} that provides satellite measurements of ozone levels
over the entire earth following the satellite tracks. We use a resolution of $0.1$ degrees in lattitude
and longitude, giving a spatial field of size $1681 \times 3601$, with over 6 million variables.
The data-set includes 172,000 measurements. To estimate the log-likelihood in presence
of missing values, we use the Schur-complement formula for determinants. Let the precision matrix
for the entire field be $J = \left( \begin{matrix} J_o & J_{o,z}\\ J_{z,o} & J_z\end{matrix} \right)$,
where subsets $\mathbf{x}_o$ and $\mathbf{x}_z$ denote the observed and unobserved components of
$\mathbf{x}$. 
%\textcolor{red}{
Then, our goal is find some parameter $\eta$ such that
$$
\max_{\eta} \int_{\mathbf{x}_z} p\left( \mathbf{x}_o , \mathbf{x}_z | \eta\right) d \mathbf{x}_z.
$$
We estimate the marginal probability using the fact that
the marginal precision matrix of $\mathbf{x}_o$ is
$\bar{J}_{o} = J_o - J_{o, z} J_z^{-1} J_{z, o}$
and
its log-determinant is computed as
$\log \det( \bar{J}_{o}) = \log \det(J) - \log \det (J_z)$ via Schur complements. 
%}
To evaluate the quadratic term
$x_o' \bar{J}_o x_o$ of the log-likelihood we need a single linear solve using an iterative solver.
We use a linear combination of the thin-plate model and the thin-membrane models
\cite{rue_GMRF}, with two parameters {$\eta = \left( \alpha,\beta\right)$}:
$J = \alpha I + \beta J_{tp} + (1-\beta) J_{tm}$ and obtain ML estimates using {\bf Algorithm \ref{alg2}}.
Note that smallest eigenvalue of $J$ is equal to $\alpha$. We show the sparse measurements in Figure \ref{fig:ozone} (a) and
the GMRF interpolation using fitted values of parameters in Figure \ref{fig:ozone} (b).
%\textcolor{red}{
We can see that the proposed log-determinant estimation algorithm allows us to do efficient estimation and inference in GMRFs of very large size, with sparse information matrices of size over 6 millions variables.
%}

\begin{figure*}[t]
\begin{center}
\includegraphics[width=\textwidth]{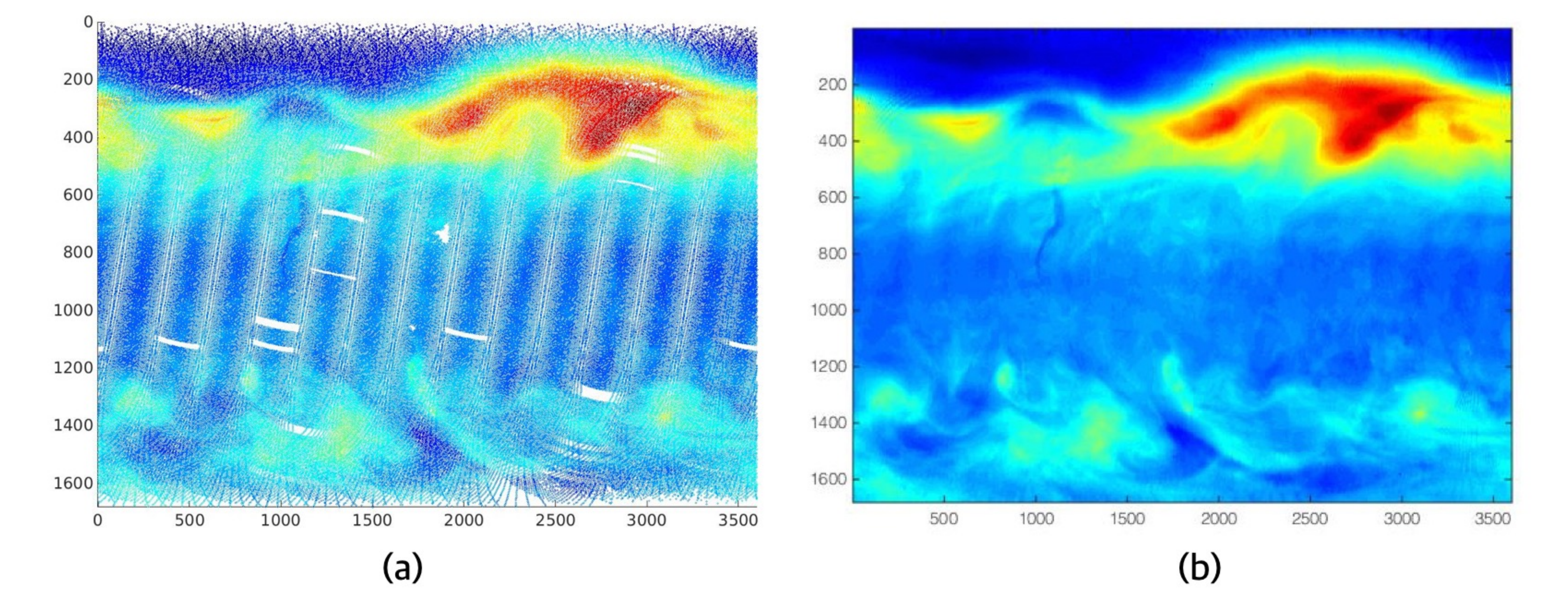}
%\vskip -0.1in
%\vspace{-0.08in}
\caption{GMRF interpolation of ozone measurements: (a) original
sparse measurements and (b) interpolated values using a GMRF with
parameters fitted using {\bf Algorithm \ref{alg2}}.}
\label{fig:ozone}
\end{center}
%\vspace{-0.15in}
\end{figure*}

\begin{figure*}[t!]
\begin{center} 
\subfigure[]{\includegraphics[width = 0.4\textwidth]{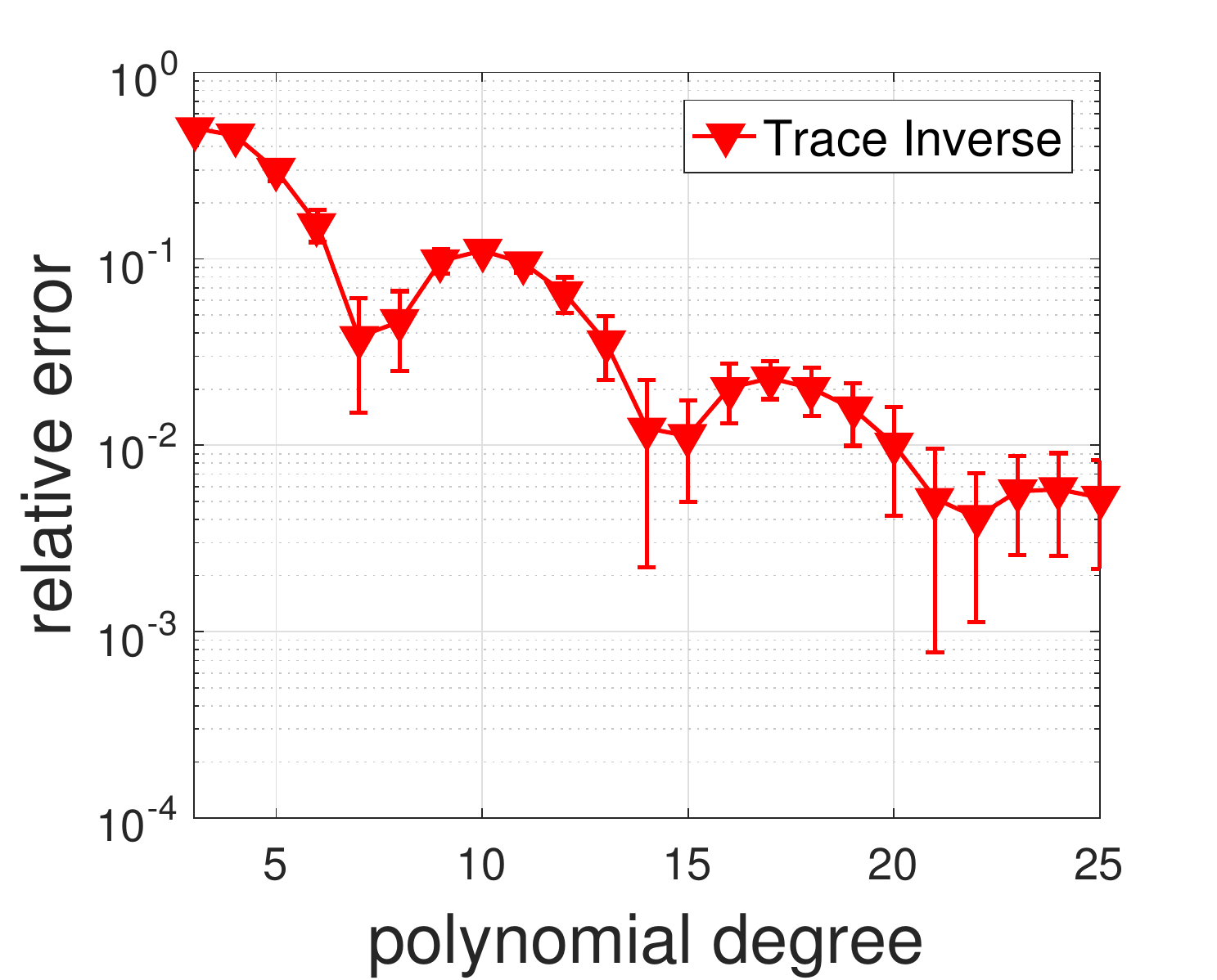}}
\subfigure[]{\includegraphics[width = 0.4\textwidth]{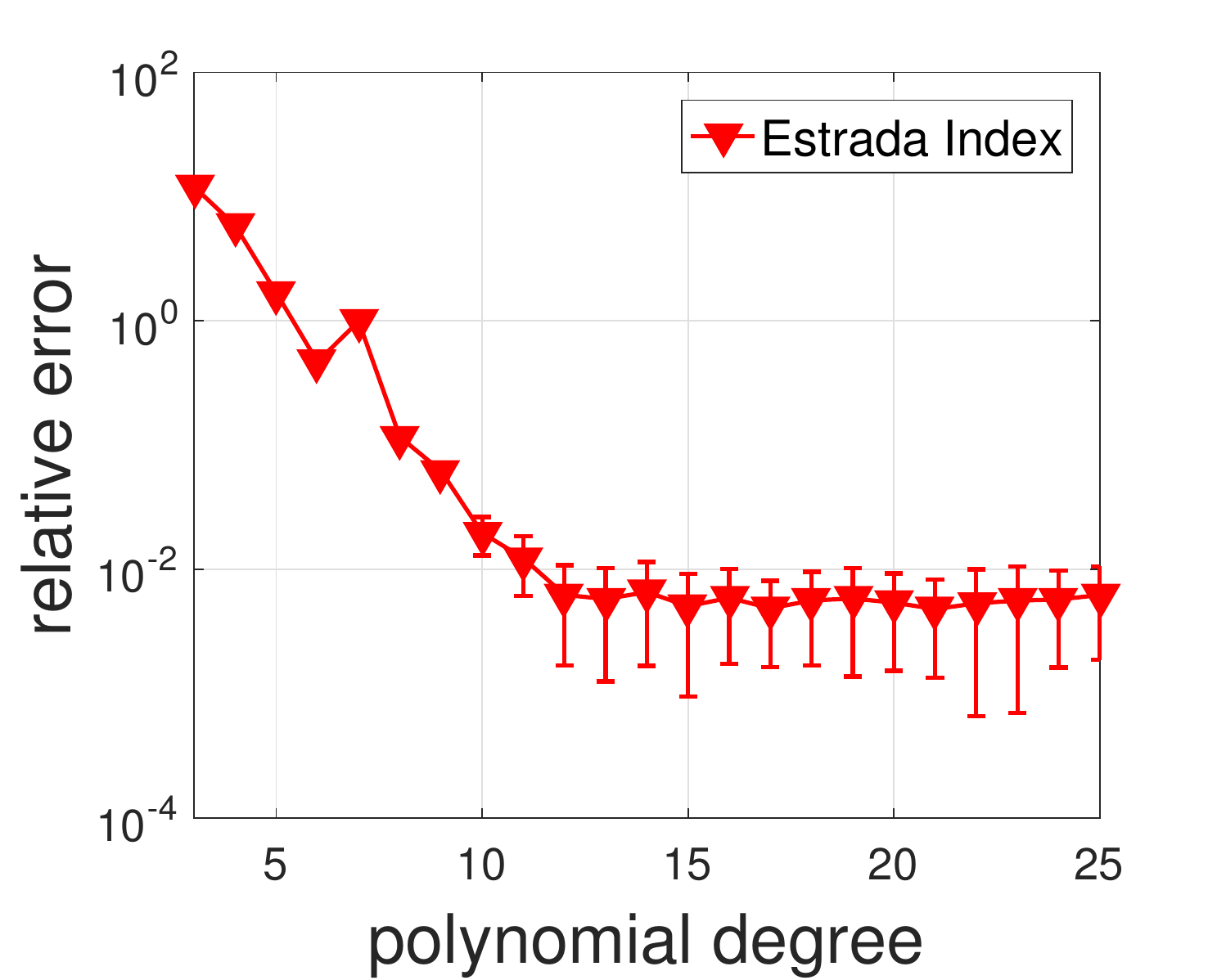}}
\vskip -0.16in
\subfigure[]{\includegraphics[width = 0.4\textwidth]{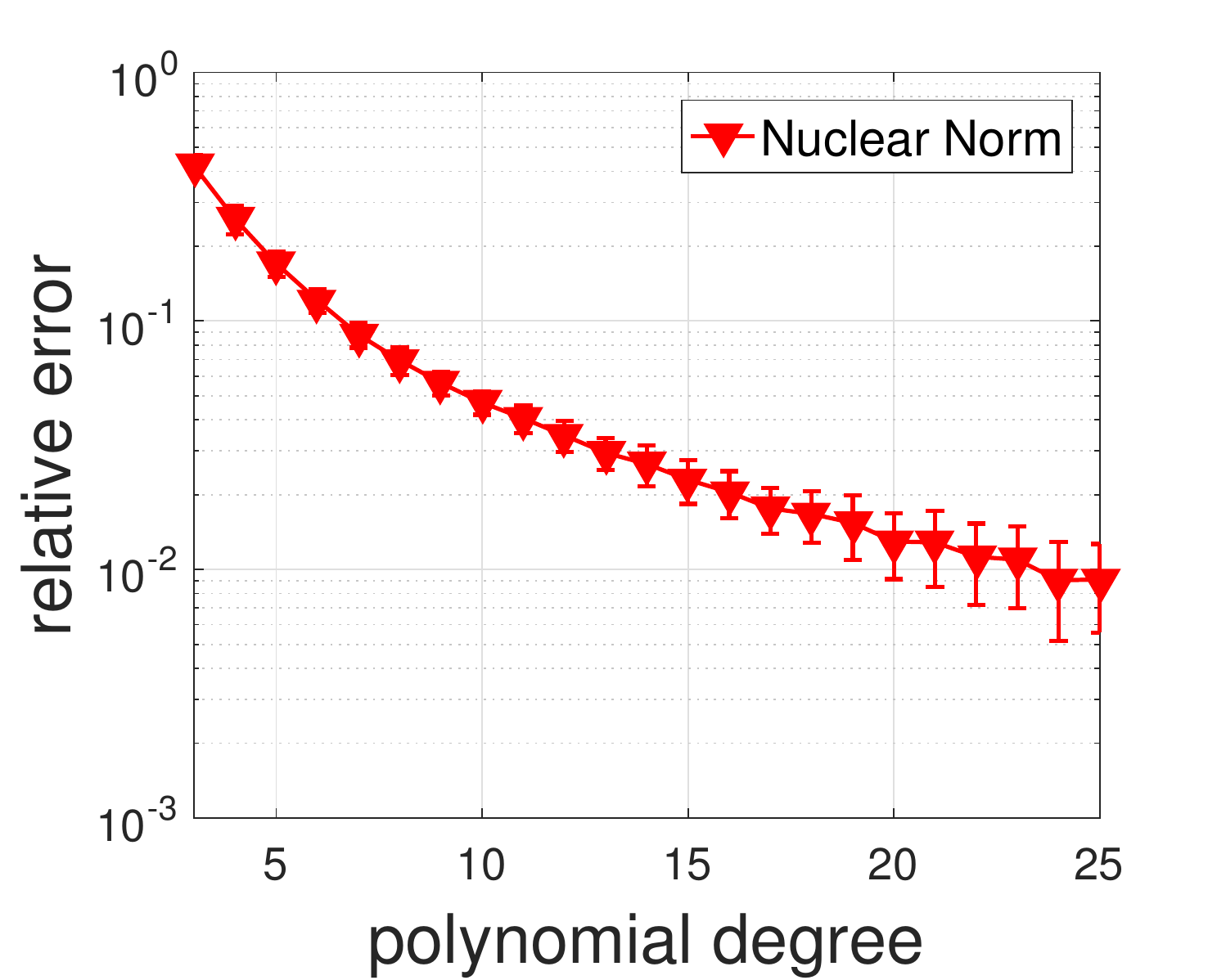}}
\subfigure[]{\includegraphics[width = 0.4\textwidth]{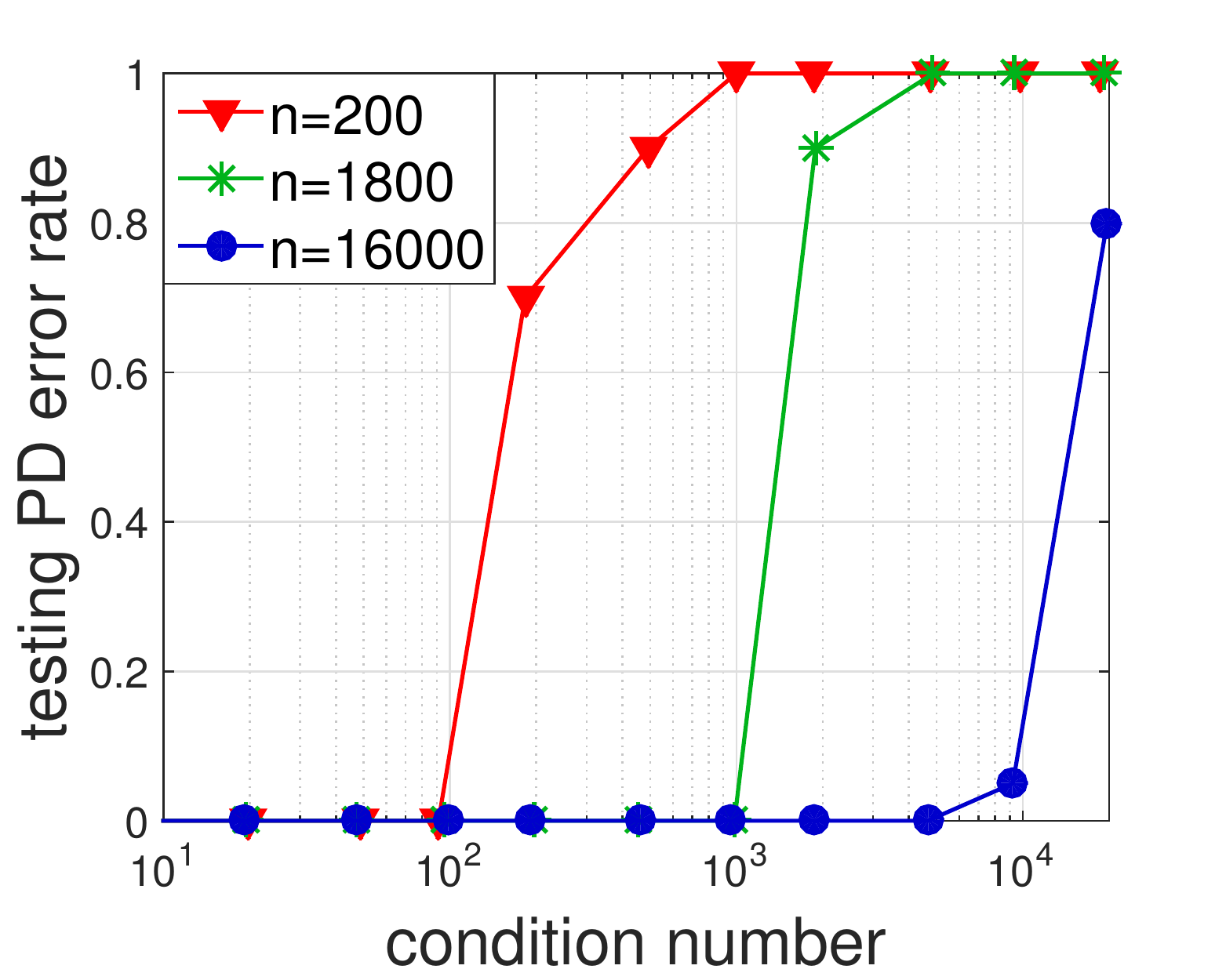}}
\caption{Accuracy of the proposed algorithm:
(a) the trace of matrix inverse,
(b) the Estrada index,
(c) the nuclear norm (Schatten $1$-norm) and
(d) testing positive definiteness.} \label{fig:application}
\end{center}
\end{figure*}

\subsection{Other Spectral Functions}
In this section, we report the performance of our scheme 
for four other choices of function $f$: the trace of matrix inverse, the Estrada index, the matrix nuclear norm
and testing positive definiteness,
which correspond to $f(x)=1/x$, $f(x)=\exp(x)$, $f(x)=x^{1/2}$ and $f(x)= \frac1{2}\left(1+\tanh\left(-\alpha x\right) \right)$, respectively.
The detailed algorithm description for each function
is given in Section \ref{sec:application}.
Since the running time of our algorithms are `almost' independent of the choice of function $f$, i.e.,
it is same as the case $f(x)=\log x$ that reported in the previous section, 
we focus on measuring the accuracy of our algorithm. %We iterate generate 100

In Figure \ref{fig:application},
we report the approximation error of our algorithm
for the trace of matrix inverse, the Estrada index, the matrix nuclear norm and testing positive definiteness.
% plots the accuracy varying polynomial degree 
All experiments were conducted on random $5000$-by-$5000$ matrices. The particular setup for
the different matrix functions are:
\begin{itemize}
    \item  The input matrix for the trace of matrix inverse is generated in the same way with 
the log-determinant case in the previous section.
%where nuclear norm is special case of Schatten $p$-norm when $p=1$.
%In each case, random matrix is generated with dimension $10000$ and the number of non-zero entries per each row is around $10$.
    \item For the Estrada index, we generate the random regular graphs with $5000$ vertices and degree $\Delta_G = 10$. 
    \item For the nuclear norm, we generate random non-symmetric matrices and estimate its nuclear norm 
(which is equal to the sum of all singular values).
We first select the $10$ positions of non-zero entries in each row and their values 
are drawn from the standard normal distribution.
The reason why we consider non-symmetric matrices is because
the nuclear norm of a symmetric matrix is much easier to compute, e.g.,
the nuclear norm of a positive definite matrix is just its trace.
We choose $\sigma_{\min}=10^{-4}$ % the lower bound of singular value sufficiently small (e.g, $a = 10^{-4}$)
and $\sigma_{\max} = \sqrt{\|A\|_1\|A\|_\infty}$ for input matrix $A$.
%Erd\H{o}s-Renyi graph with $10000$ vertices and average vertex degree 10. 
%for every vertex 
%so that adjaceny matrix has the same sparsity of previouse setup. 
%The output of our algorithm is more accurate as degree increases.
%We observe the estimated relative error of our algorithm is at most $1\%$
%for these spectral functions.
%decreases as polynomial degree grows.
%\textcolor{red}{
    \item For testing positive definiteness,
%we generate positive definite matrices whose condition number changes from $10$ to $20000$.
we first create random symmetric matrices whose the smallest eigenvalue varies from $10^{-1}$ to $10^{-4}$ and 
the largest eigenvalue is less than 1 (via appropriate normalizations).
Namely, the condition number is between $10$ and $10^4$.
%and shift their eigenvalues to specific interval, e.g., $\left[ 10^{-2}, 1\right]$.
We choose the same sampling number $m=50$
and three different polynomial degrees: $n = 200, 1800$ and $16000$.
For each degree $n$,
{\bf Algorithm \ref{alg:psd}} detects correctly positive definiteness of matrices
with condition numbers at most $10^2$, $10^3$ and $10^4$, respectively.
The error rate is measured as a ratio of incorrect results among 20 random instances.
%}
\end{itemize}

\iffalse
In order to test positive definiteness, we generate random symmetric matrices 
whose smallest eigenvalues are in $\left[-0.025, 0.025 \right]$.
%Then, we iterate {\bf Algorithm \ref{alg:psd}} 100 times to decide whether input matrix is positive definite or not. 
The error rate indicates the ratio of the number of correct guess to that of total tests.
Figure \ref{fig:application} (d) shows that our algorithm is always correct if the absolute value of
the smallest eigenvalue is larger than $0.02$ under the setting of $m=50$ and $n=25$.
%the result of experiment and matrix with extremely small eigenvalues 
%is hard to guess of its positive definiteness.
%}
\fi

%Figure \ref{fig:application} (d) shows that proposed algorithm is always correct 
%if the condition number smallest eigenvalue is larger than $0.02$ under the setting of $m=50$ and $n=25$.

%\textcolor{red}{
For experiment of the trace of matrix inverse, the Estrada index and the nuclear norm, 
we plot the relative error of the proposed algorithms varying polynomial degrees in 
Figure \ref{fig:application} (a), (b) and (c), respectively.
Each of them achieves less than 1$\%$ error with polynomial degree at most $n=25$ and sampling number $m=50$.
Figure \ref{fig:application} (d) shows the results of testing positive definiteness.
When $n$ is set according to the condition number the proposed algorithm is almost always 
correct in detecting positive definiteness.
For example, if the decision problem involves with the active region $\cal{A}_{\varepsilon}$ for $\varepsilon = 0.02$,
which is the case that matrices having the condition number at most 100,
polynomial degree $n=200$ is enough for the correct decision.

\begin{table}[h]
\centering
\scriptsize
%\footnotesize
\begin{tabular}{@{}ccccccccc@{}}
\toprule
matrix & dimension & \begin{tabular}[c]{@{}c@{}}number of\\ nonzeros\end{tabular} & \begin{tabular}[c]{@{}c@{}}positive\\ definite\end{tabular} & \begin{tabular}[c]{@{}c@{}}{\bf Algorithm \ref{alg:psd}}\\ $n=200$\end{tabular} & \begin{tabular}[c]{@{}c@{}}{\bf Algorithm \ref{alg:psd}}\\ $n=1800$\end{tabular} & \begin{tabular}[c]{@{}c@{}}{\bf Algorithm \ref{alg:psd}}\\ $n=16000$\end{tabular} & \begin{tabular}[c]{@{}c@{}} MATLAB \\ ${\tt eigs}$ \end{tabular} & \begin{tabular}[c]{@{}c@{}} MATLAB \\ ${\tt condest}$ \end{tabular}\\ \midrule
${\tt Chem97ZtZ}$	    & 2,541	     & 7,361      & yes       & PD        & PD        & PD      & diverge  & 462.6\\
% ${\tt Muu}$	            & 7,102	     & 170,134    & yes       & PD        & PD
${\tt fv1}$     	    & 9,604	     & 85,264     & yes       & PD        & PD        & PD      & 0.5122  & 12.76\\
${\tt fv2}$	            & 9,801	     & 87,025     & yes       & PD        & PD        & PD      & 0.5120  & 12.76\\
${\tt fv3}$             & 9,801	     & 87,025     & yes       & NOT PD    & NOT PD    & PD      & 0.0020  & 4420\\
${\tt CurlCurl\_0}$	    & 11,083	 & 113,343    & no        & NOT PD    & NOT PD    & NOT PD  & diverge  & 
6.2%6.2491 
{$\times 10^{21}$}\\
${\tt barth5}$          & 15,606     & 107,362    & no        & NOT PD    & NOT PD    & NOT PD  & -2.1066  & 84292\\
${\tt Dubcova1}$        & 16,129     & 253,009    & yes       & NOT PD    & NOT PD    & PD      & 0.0048   & 2624\\
${\tt cvxqp3}$	        & 17,500	 & 114,962    & no        & NOT PD    & NOT PD    & NOT PD  & diverge  & 
2.2%2.1701
{$\times 10^{16}$}\\
${\tt bodyy4}$          & 17,546     & 121,550    & yes       & NOT PD    & NOT PD    & PD      & diverge  & 1017\\
${\tt t3dl{\_}e}$	    & 20,360	 & 20360	  & yes       & NOT PD    & NOT PD    & PD      & diverge  & 6031\\
${\tt bcsstm36}$        & 23,052     & 320,060    & no        & NOT PD    & NOT PD    & NOT PD  & diverge  & $\infty$\\
${\tt crystm03}$        & 24,696     & 583,770    & yes       & NOT PD    & PD        & PD      & 
3.7%3.7072
{$\times 10^{-15}$} & 467.7\\
${\tt aug2d}$           & 29,008     & 76,832     & no        & NOT PD    & NOT PD    & NOT PD  & -2.8281  & $\infty$\\
${\tt wathen100}$	    & 30,401	 & 471,601    & yes       & NOT PD    & NOT PD    & PD      & 0.0636   & 8247\\
${\tt aug3dcqp}$	    & 35,543	 & 128,115    & no        & NOT PD    & NOT PD    & NOT PD  & diverge  & 
4.9%4.909 
{$\times 10^{15}$}\\
${\tt wathen120}$	    & 36,441	 & 565,761	  & yes       & NOT PD    & NOT PD    & PD      & 0.1433   & 4055\\
${\tt bcsstk39}$	    & 46,772	 & 2,060,662  & no        & NOT PD    & NOT PD    & NOT PD  & diverge  & 
3.1%3.0797
{$\times 10^{8}$}\\
${\tt crankseg{\_}1}$	& 52,804	 & 10,614,210 & yes       & NOT PD    & NOT PD    & NOT PD  & diverge  & 
2.2%2.2230
{$\times 10^{8}$}\\
${\tt blockqp1}$        & 60,012     & 640,033    & no        & NOT PD    & NOT PD    & NOT PD  & -446.636 & 8.0{$\times 10^{5}$}\\
${\tt Dubcova2}$	    & 65,025	 & 1,030,225  & yes       & NOT PD    & NOT PD    & PD      & 0.0012   & 10411\\
% ${\tt obstclae}$	    & 40,000	 & 197,608    & yes       & PD        & PD        & PD      & 1.96 \\
% ${\tt jnlbrng1}$        & 40,000     & 199,200    & yes       & NOT PD    & PD        & PD      & 0.1006\\
% ${\tt bfly}$            & 49,152	 & 196,608    & no        & NOT PD    & NOT PD    & NOT PD  & -4.000\\
% ${\tt finan512}$        & 74,752     & 596,992    & yes       & PD        & PD        & PD      & 0.9475\\
% ${\tt shallow{\_}water1}$ & 81,920   & 327,680    & yes       & PD        & PD        & PD      & 5.7897{$\times 10^9 $}\\
${\tt thermomech{\_}TC}$ & 102,158   & 711,558    & yes       & NOT PD    & PD        & PD      & 0.0005   & 125.5\\
${\tt Dubcova3}$	    & 146,689	 & 3,636,643  & yes       & NOT PD    & NOT PD    & PD      & 0.0012   & 11482\\
${\tt thermomech{\_}dM}$ & 204,316   & 1,423,116  & yes       & NOT PD    & PD        & PD      & 
9.1%9.0794
{$\times 10^{-7} $}  & 125.487\\
% ${\tt c{\mbox{-}}73}$	& 169,422	 & 1,279,274  & no        & NOT PD    & NOT PD    & NOT PD  & -698.1\\
${\tt pwtk}$            & 217,918    & 11,524,432 & yes       & NOT PD    & NOT PD    & NOT PD  & diverge  & 
5.0%5.0348
{$\times 10^{12}$}\\ 
${\tt bmw3{\_2}}$	    & 227,362	 & 11,288,630 & no        & NOT PD    & NOT PD    & NOT PD  & diverge  &
1.2%1.1516 
{$\times 10^{20}$}\\
\bottomrule
\end{tabular}
\caption{Testing positive definiteness for real-world matrices. 
{\bf Algorithm \ref{alg:psd}} outputs PD or NOT PD, i.e., the input matrix is either
(1)  positive definite (PD) or
(2) {not positive definite} or its smallest eigenvalue is in the indifference region (NOT PD).
The MATLAB ${\tt eigs}$ and 
${\tt condest}$ functions output the smallest eigenvalue and an estimate for the condition number of
the input matrix, respectively.}
\label{table:collection}
\end{table}

%\textcolor{red}{
We tested the proposed algorithm for testing positive definiteness on 
real-world matrices from the University of Florida Sparse Matrix Collection \cite{davis2011university},
selecting various symmetric %positive definite and indefinite 
matrices.
We use $m=50$ and three choices for $n$: $n=200,1800,16000$.
The results are reported in Table \ref{table:collection}.
%The algorithm with these settings detects positive definiteness
%almost exactly for matrices with the condition number up to $10^4$.
%even for huge dimension.
%if dimension becomes huge.
We observe that the algorithm is always correct when declaring positive definiteness, but seems to
declare indefiniteness when the matrix is too ill-conditioned for it to detect definiteness correctly.
In addition, with two exceptions (${\tt crankseg{\_}1}$ and ${\tt pwtk}$), when $n=16000$ the algorithm was correct in declaring whether the 
matrix is positive definite or not. We remark that while $n=16000$ is rather large it is still smaller
than the dimension of most of the matrices that were tested (recall that our goal was to develop an 
algorithm that requires a small number of matrix products, i.e., it does not grow with respect to the matrix dimension).
We also note that even when the algorithm fails it still provides useful information
%gives an important hint 
about both positive definiteness and the condition number of an input matrix 
while standard methods such as Cholesky decomposition (as mentioned in Section \ref{sec:psd}) are intractable for large matrices.
Furthermore, one can first run an algorithm to estimate the condition number, e.g.,
the MATLAB ${\tt condest}$ function, and then choose an appropriate degree $n$.
%which is intractable for large-scale matrices due to its cubic-time complexity.
%\textcolor{red}{
%To choose proper $n$, we need prior knowledge of the condition numbers. 
%In MATLAB, ${\tt condest}$ function estimates the upper bound of condition number of the matrix with iterative search.
%This can find 1-
%}
We also run the MATLAB ${\tt eigs}$ function 
which is able to estimate the smallest eigenvalue using iterative methods \cite{ipsen1997computing}
(hence, it can be used for testing positive definitesss).
Unfortunately, the iterative method 
%implemented in ${\tt eigs}$ function 
often does not converge, i.e, residual tolerance may not go to zero,
as reported in Table \ref{table:collection}.
%As reported in the table, 
%some of the test matrices can be detected positive definiteness under both our proposed algorithm and ${\tt eigs}$ function,
%but the most of others diverge to estimate the smallest eigenvalues using ${\tt eigs}$ function.
%In summary, our proposed algorithm can distinguish positive definiteness and the upper bound of condition number.
One advantage of our algorithm is that it does not depend on a convergence criteria 

\section{Conclusion}
Recent years has a seen a surge in the need for various computations on large-scale unstructured matrices.
The lack of structure poses a significant challenge for traditional decomposition based methods. 
Randomized methods are a natural candidate for such tasks as they are mostly oblivious to structure. 
In this paper, we proposed and analyzed a linear-time approximation algorithm for spectral sums of symmetric matrices, where the exact computation requires cubic-time in the worst case. 
Furthermore, our algorithm is very easy to parallelize since it requires only (separable) 
matrix-vector multiplications. 
{
We believe that the proposed algorithm will find 
an important theoretical and computational roles in a variety of applications ranging from statistics
and machine learning to applied science and engineering.
%numerous applications in machine learning problems.
}

\subsection*{Acknowledgement}
Haim Avron acknowledges the support from the XDATA program of the Defense Advanced Research Projects
Agency (DARPA), administered through Air Force Research Laboratory contract FA8750-12-C-0323. The authors thank Peder Olsen and Sivan Toledo for helpful discussions. 

\bibliography{biblist}
\bibliographystyle{apalike}

\end{document}